\documentclass[10pt, onecolumn]{IEEEtran}
\usepackage{graphicx}
\usepackage{amsmath}
\usepackage{amssymb}
\usepackage{amsfonts}
\usepackage{epsfig}

\usepackage{amsmath,mathtools}    % need for subequations
\usepackage{graphicx}   % need for figures
\usepackage{color}      % use if color is used in text
\usepackage{hyperref}
\usepackage{comment}
\usepackage{subfigure}
\usepackage{times}
\usepackage{hyphenat}
\usepackage{epsfig}
\usepackage{latexsym}
\usepackage{bbold,bbm}
\usepackage{setspace}
\usepackage{amssymb,amsmath}
\usepackage{mathrsfs}
\usepackage{amsgen,amsfonts,amsbsy,amsthm}
\usepackage{algorithmic,algorithm}
\usepackage{array}
\usepackage{booktabs}
\usepackage{amsfonts}
\usepackage{amssymb}
\usepackage{multirow}
\usepackage{longtable}
\usepackage{epstopdf}
\usepackage{wrapfig}
\usepackage[font=small]{caption}
\usepackage{tikz}
\usetikzlibrary{shapes,arrows}

\newtheorem{defn}{Definition}
\newtheorem{thm}{Theorem}

\newtheorem{cor}[thm]{Corollary}
\newtheorem{note}{Remark}

\newtheorem{claim}{Claim}

\newcommand{\bit}{\begin{itemize}}
	\newcommand{\eit}{\end{itemize}}
\newcommand{\bcor}{\begin{cor}}
	\newcommand{\ecor}{\end{cor}}
\newcommand{\beq}{\begin{equation}}
\newcommand{\eeq}{\end{equation}}
\newcommand{\beqn}{\begin{equation*}}
\newcommand{\eeqn}{\end{equation*}}
\newcommand{\bea}{\begin{eqnarray}}
\newcommand{\eea}{\end{eqnarray}}
\newcommand{\bean}{\begin{eqnarray*}}
	\newcommand{\eean}{\end{eqnarray*}}
\newcommand{\ben}{\begin{enumerate}}
	\newcommand{\een}{\end{enumerate}}
\newcommand{\bdefn}{\begin{defn}}
	
	\renewcommand\footnotemark{}

			 	\newcommand{\uh}{\underline{h}}

			\newcommand{\calc}{\ensuremath{\mathcal{C}}}
			\newcommand{\calcp}{\ensuremath{\mathcal{C}^{\perp}}}

	\newcommand\scalemath[2]{\scalebox{#1}{\mbox{\ensuremath{\displaystyle #2}}}}
	
\begin{document}
			\sloppy
			\title{A Tight Rate Bound and a Matching Construction for Locally Recoverable Codes with Sequential Recovery From Any Number of Multiple Erasures}
			
			\author{
				\IEEEauthorblockN{S. B. Balaji, Ganesh R. Kini and P. Vijay Kumar, \it{Fellow}, \it{IEEE}}
				
				\IEEEauthorblockA{Department of Electrical Communication Engineering, Indian Institute of Science, Bangalore.  \\ Email: balaji.profess@gmail.com, kiniganesh94@gmail.com, pvk1729@gmail.com} 
				
				\thanks{P. Vijay Kumar is also an Adjunct Research Professor at the University of Southern California.  This research is supported in part by the National Science Foundation under Grant 1421848 and in part by an India-Israel UGC-ISF joint research program grant.} 
				\thanks{S. B. Balaji would like to acknowledge the support of TCS research-scholarship program.}
			}
			\maketitle
			
			\begin{abstract}
               An $[n,k]$ code $\mathcal{C}$ is said to be locally recoverable in the presence of a single erasure, and with locality parameter $r$, if each of the $n$ code symbols of $\mathcal{C}$ can be recovered by accessing at most $r$ other code symbols.  An $[n,k]$ code is said to be a locally recoverable code with sequential recovery from $t$ erasures, if for any set of $s \leq t$ erasures, there is an $s$-step sequential recovery process, in which at each step, a single erased symbol is recovered by accessing at most $r$ other code symbols.   This is equivalent to the requirement that for any set of $s \leq t$ erasures, the dual code contain a codeword whose support contains the co-ordinate of precisely one of the $s$ erased symbols.  
	
	           In this paper, a tight upper bound on the rate of such a code, for any value of number of erasures $t$ and any value $r \geq 3$, of the locality parameter is derived.  This bound proves an earlier conjecture due to Song, Cai and Yuen.   While the bound is valid irrespective of the field over which the code is defined, a matching construction of {\em binary} codes that are rate-optimal is also provided, again for any value of $t$ and any value $r\geq3$. 
			\end{abstract}
			
			%\tableofcontents
			
			\begin{IEEEkeywords} Distributed storage, locally recoverable codes, codes with locality, locally repairable codes, sequential repair, multiple erasures, rate bound, proof of conjecture.
			\end{IEEEkeywords}
			
%			\section{Introduction}
%			An $[n,k]$ code is said to have locality $r$ if each of the $n$ code symbols of $\mathcal{C}$ can be recovered by accessing at most $r$ other code symbols. Equivalently, there exist $n$ codewords ${h_1 \cdots h_n}$ in the dual code $\mathcal{C}^\perp$ such that $c_i \in \text{supp}(h_i)$ and $|\text{supp}(h_i)| \leq r+1$ for $1 \leq i \leq n$ where $c_i$ denote the $i^\text{th}$ code symbol of $\mathcal{C}$ and $\text{supp}(h_i)$ denote the support of the codeword $h_i$.
%			\paragraph{Codes with Sequential Recovery}
%			An $[n,k]$ code is defined as a code with sequential recovery \cite{BalPraKum} from $t$ erasures having locality $r$ if for any set of $s \leq t$ erased symbols, $\{c_{\sigma_1},...,c_{\sigma_s} \}$, there exists a codeword $h$ in the dual of the code, of Hamming weight $\leq r+1$, such that $\text{supp}(h) \cap \{\sigma_1,...,\sigma_s \} = 1$.  The parameter $r$ is the locality parameter and we will formally refer to this class of codes as $(n,k,r,t)_{\text{seq}}$ codes. When the parameters $(n,k,r,t)$ are clear from the context, we will simply refer to a code in this class as a code with sequential recovery.
			\section{Introduction}
			An $[n,k]$ code $\mathcal{C}$ is said to have locality $r$ if each of the $n$ code symbols of $\mathcal{C}$ can be recovered by accessing at most $r$ other code symbols. Equivalently, there exist $n$ codewords ${\uh_1 \cdots \uh_n}$, not neccessarily distinct, in the dual code $\mathcal{C}^\perp$, such that $i \in \text{supp}(\uh_i)$ and $|\text{supp}(\uh_i)| \leq r+1$ for $1 \leq i \leq n$ where $\text{supp}(\uh_i)$ denotes the support of the codeword $\uh_i$.
			\paragraph{Codes with Sequential Recovery}
			An $[n,k]$ code \calc\ over a field $\mathbb{F}_q$ is defined as a code with sequential recovery \cite{BalPraKum} from $t$ erasures and with locality-parameter $r$, if for any set of $s \leq t$ erased symbols $\{c_{\sigma_1},...,c_{\sigma_s} \}$, there exists a codeword $\uh$ in the dual code \calcp\ of Hamming weight $\leq r+1$, such that $\text{supp}(\uh) \cap \{\sigma_1,...,\sigma_s \} = 1$.  We will formally refer to this class of codes as $(n,k,r,t)_{\text{seq}}$ codes. When the parameters $(n,k,r,t)$ are clear from the context, we will simply refer to a code in this class as a code with sequential recovery.
			\subsection{Background}
			In \cite{GopHuaSimYek} Gopalan et al. introduced the concept of codes with locality (see also \cite{PapDim,OggDat}), where an erased code symbol is recovered by accessing a small subset of other code symbols. The size of this subset denoted by $r$, is typically much smaller than the dimension of the code, making the repair process more efficient when compared with MDS codes. The focus of \cite{GopHuaSimYek} was local recovery from single erasure (see also \cite{GopHuaSimYek,HuaChenLi,KamPraLalKum,TamBar_Optimal_LRC}).
			
			The sequential approach to recovery from erasures, introduced by Prakash et al. \cite{PraLalKum} is one of several  approaches to locally recover from multiple erasures. Codes employing this approach have been shown to be better in terms of rate and minimum distance (see \cite{PraLalKum,RawMazVis,SonYue_3_Erasure,SonYue_Binary_local_repair,BalPraKum,balaji2016binary}).  Local recovery in the presence of two erasures is considered in \cite{PraLalKum} (see also \cite{SonYue_3_Erasure}) where a tight rate bound for two erasure case and an optimal construction is provided. Codes with sequential recovery from three erasures can be found discussed in \cite{SonYue_3_Erasure,BalPraKum,song2016sequential}. A bound on rate of an $(n,k,r,3)_{seq}$ code was derived in \cite{song2016sequential}.   A rate bound for $t=4$ appears in \cite{balaji2016binary}. 
			
			There are several approaches to local recovery from multiple erasures: 
%			\vspace{-0.1cm}
			\ben
			\item {\em Stronger Local Codes:} Here, every code symbol is contained in a local code of length at most $r+t$ and minimum distance at least $t+1$, see (\cite{PraKamLalKum}). While the recovery process is sequential, the erased symbols can be recovered in arbitrary order, see \cite{KamPraLalKum,SonDauYueLi}.
			\item {\em Codes with Availability:} For every code symbol $c_i$ in $\mathcal{C}$, there exist $t$ codewords $\underline{h}^i_1,...,\underline{h}^i_t$ in the dual of the code, each of Hamming weight $\leq r+1$, such that $\text{supp}(\underline{h}^i_g) \cap \text{supp}(\underline{h}^i_j) = \{i\} $, $\forall\ 1 \leq g \neq j \leq t$ ; recovery can be carried out in parallel. For more details see \cite{TamBarFro,WanZha_Combinatorial_Repair_locality,TamBar_Optimal_LRC,WangZhang_multiple_erasure}.
			%RawPapDimVis_arxiv
			\item {\em Codes with Selectable Recovery:} Here, given any set of $t$ erasures, every erased symbol has a parity check that involves that symbol and no other erased symbol; with these codes, one is free to choose the order in which to recover the erased symbols; recovery in parallel may or may not be possible, depending upon the specific construction. For details, see \cite{JuaHolOgg}.
			\item {\em Codes with Co-operative local recovery:} A code with co-operative local recovery has the property that for a given $t$, any set of $t$ erased symbols can be co-operatively recovered by accessing at most $r$ other symbols, see \cite{RawMazVis}.
			\item {\em Codes with Sequential Recovery:}  This class of codes is the object of study in the present paper and a definition has been provided above.  The class of sequential-recovery codes is a larger class of codes that contains all the four above-mentioned classes of codes as depicted in Fig.~\ref{fig:codeclasses}.  For this reason, codes with sequential recovery can potentially achieve higher rate and have larger minimum distance.
			\een
			%	\begin{center}
			%\begin{figure}
			%\includegraphics[trim= 0in  0in 0in 0in, width=4.0in]{codeclasses.pdf}%trim l b r t 
			%
			%%\includegraphics[width=4.3in]{multiple_erasures_overview_bits.pdf}
			%\end{figure}
			%\end{center}
			
			\begin{figure}[ht]
				\centering
				\includegraphics[scale = 0.25]{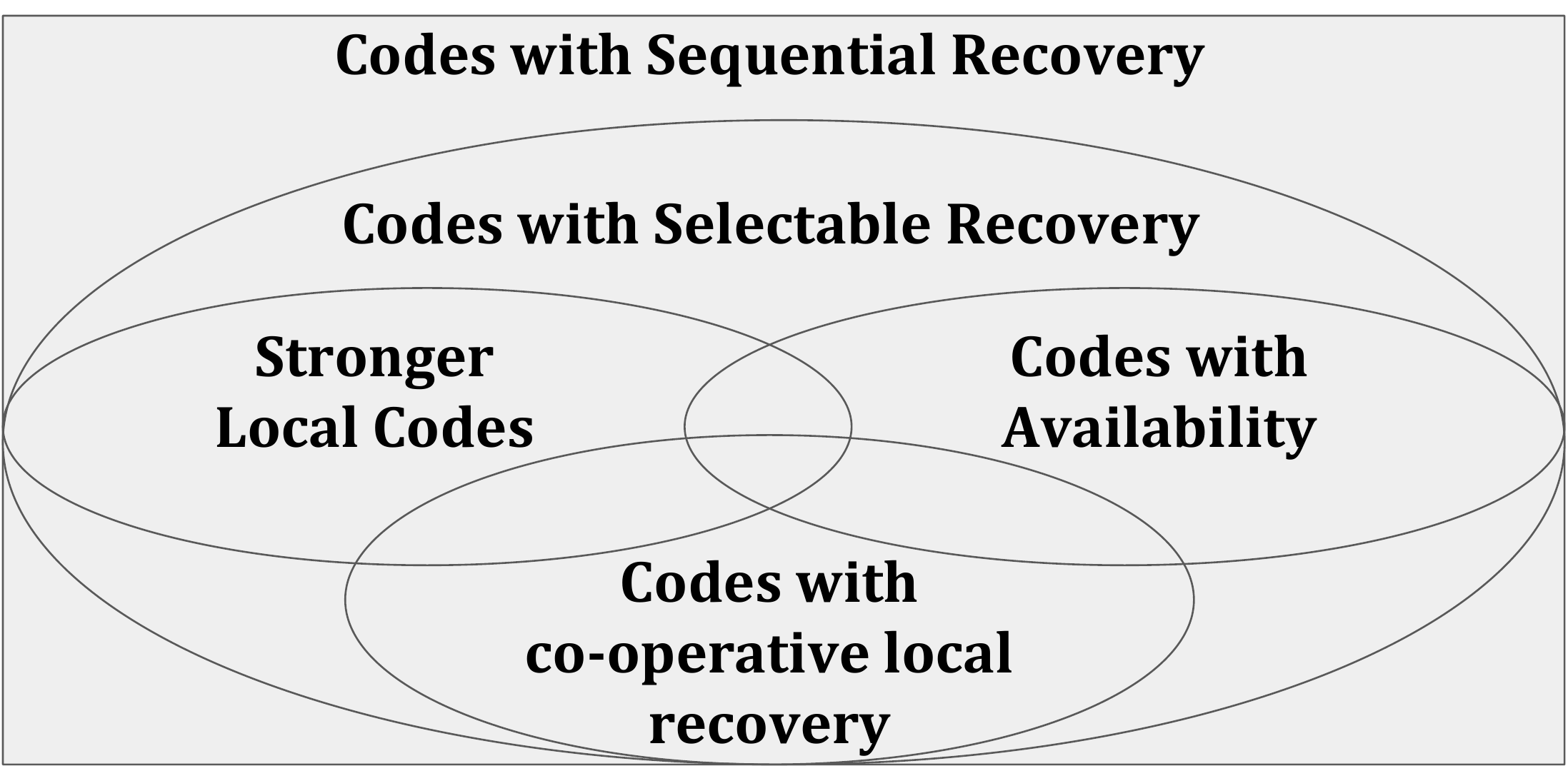}
				\caption{Figure showing the relationship of codes with sequential recovery to some other classes of codes designed for recovery from multiple erasures. }
				\label{fig:codeclasses}
			\end{figure}
			%	\begin{figure}[ht]
			%		\centering
			%		\includegraphics[scale = 0.26]{multiple_erasures_overview.pdf}
			%		\caption{Figure showing that codes with sequential recovery contain certain other important classes of codes.}
			%		\label{fig:codeclasses}
			%	\end{figure}
			%	
%			\vspace{-0.6cm}
			\subsection{Contributions of the Paper} 
			In this paper, we derive an upper bound on the rate of a code having locality-parameter $r$ with sequential recovery from $t$ erasures, for any $r \geq 3$ and any $t$. While the bound is valid irrespective of the field over which the code is defined, we provide here a matching construction of {\em binary} codes that are rate-optimal, i.e., binary codes achieving the rate bound for any $t$ and any $r \geq 3$.   These results are also shown to prove a conjecture on code rate due to Song, Cai and Yeun \cite{SonYue_3_Erasure}. 
						
			\section{Upper Bound on Rate of an $(n,k,r,t)_{seq}$ code}
			In this section we provide an upper bound on the rate of an $(n,k,r,t)_{seq}$ code for $r \geq 3$. The cases of $t$ even and $t$ odd are considered in two separate subsections.
			
			\subsection{Upper Bound on Rate of an $(n,k,r,t)_{seq}$ Code for $t$ Even:}
			In this subsection we provide an upper bound on the rate of an $(n,k,r,t)_{seq}$ code for $t$ even and $r \geq 3$. 
			
			\begin{thm}\label{rate_teven}
				Let $\mathcal{C}$ be an $(n,k,r,t)_{seq}$ code over a field $\mathbb{F}_q$. Let $t$ be a positive even integer and $r \geq 3$. Then 
				
				\begin{equation}
				\frac{k}{n} \leq \frac{r^{\frac{t}{2}}}{r^{\frac{t}{2}} + 2 \sum\limits_{i=0}^{\frac{t}{2}-1} r^i} \label{Thm1}.
				\end{equation}
				
			\end{thm}
			
			\begin{proof}
				
				We begin by setting $\mathcal{B}_0 = \text{span}({\underline{c} \in \mathcal{C}^{\perp} : w_H(\underline{c}) \leq r+1})$ ($w_H(\underline{c})$ denotes the Hamming weight of the vector $\underline{c}$).
				Let $m$ be the dimension of $\mathcal{B}_0$. Let ${\underline{c}_1,...\underline{c}_m}$ (row vectors) be a basis of $\mathcal{B}_0$ such that $w_H(\underline{c}_i) \leq r+1$.
				\bean
				\text{Let } H_1  & = &  \left[ 
				\begin{array}{c}
					\underline{c}_1 \\
					\vdots \\
					\underline{c}_m			
				\end{array}\right].
				\eean
				It follows that $H_1$ is a parity check matrix of an $(n,n-m,r,t)_{seq}$ code as its row space contains all the codewords of Hamming weight at most $r+1$ which are contained in $\mathcal{C}^\perp$. Also, 
				\bean
				\frac{k}{n} \leq 1 - \frac{m}{n}.
				\eean
				 The idea behind the next few arguments in the proof is the following. The codes with largest rate will tend to have a larger value of $n$ for fixed $m$.  On the other hand, the Hamming weight of the matrix $H_1$ (i.e., the number of non-zero entries in the matrix) is bounded above by $m(r+1)$.  	 It follows that to make $n$ large, one would like the columns of $H_1$ to have as small a weight as possible. It is therefore quite natural to start building $H_1$ by picking as many columns of weight $1$ as possible, then columns of weight $2$ and so on.  As one proceeds by following this approach, it turns out that the matrix $H_1$ is forced to have a certain sparse, block-diagonal, staircase form and an understanding of this structure is used to derive the upper bound on code rate. 
				We now proceed to derive an upper bound on $1 - \frac{m}{n}$.
				It can be seen that the matrix $H_1$, after permutation of rows and columns can be written in the following form
				\bea
				H_1 = \left[
				\scalemath{1}{
					\begin{array}{c|c|c|c|c|c|c|c|c}
						D_0 & A_1 & 0 & 0 & \hdots & 0 & 0 & 0 & \\
						\cline{1-8}
						0 & D_1 & A_2 & 0 & \hdots & 0 & 0 & 0 & \\
						\cline{1-8}
						0 & 0 & D_2 & A_3 & \hdots & 0 & 0 & 0 & \\
						\cline{1-8}
						0 & 0 & 0 & D_3 & \hdots & 0 & 0 & 0 &  \\
						\cline{1-8}
						\vdots & \vdots & \vdots & \vdots & \hdots & \vdots & \vdots & \vdots & D\\
						\cline{1-8}
						0 & 0 & 0 & 0 & \hdots & A_{\frac{t}{2}-2} & 0 & 0 & \\
						\cline{1-8}
						0 & 0 & 0 & 0 & \hdots & D_{\frac{t}{2}-2} & A_{\frac{t}{2}-1} & 0 & \\
						\cline{1-8}
						0 & 0 & 0 & 0 & \hdots & 0 & D_{\frac{t}{2}-1} &  & \\
						\cline{1-7}
						0 & 0 & 0 & 0 & \hdots & 0 & 0 & C & \\
					\end{array} 
				}
				\right], \ \ \label{Hform}
				\eea
%				
%				where 
%				\begin{itemize}
%					\item the rows of $H_1$ are labeled ${1,2,...,m}$ and columns are labeled ${1,2,...n}$,
%					\item the matrices $\{A_i, 1\leq i \leq \frac{t}{2}-1\}$ have dimensions $\rho_{i-1} \times a_i$, the matrices $\{D_i, 0 \leq i \leq \frac{t}{2}-1\}$ have dimensions $\rho_{i} \times a_i$,
%					\item $D_0$ is a $\rho_0 \times a_0$ matrix with each column having Hamming weight 1 and each row having Hamming weight at least 1,
%					\item  $\{A_i\}$,$\{D_i\}$,$\{B_i:= \left[\frac{A_i}{D_i}\right]\}$ are matrices such that each column of $B_i$ has Hamming weight 2, each column of $A_i$ has Hamming weight at least 1 and each row of $D_i$ has Hamming weight at least 1,
%					\item $C$ is a matrix with each column having Hamming weight 2,
%					\item $D$ is a matrix with each column having Hamming weight at least 3.
%				\end{itemize}
				where 
				\begin{enumerate}
					\item the rows of $H_1$ are labeled ${1,2,...,m}$ and columns are labeled ${1,2,...n}$,
					\item $A_i$ is a $\rho_{i-1} \times a_i$ matrix for $1\leq i \leq \frac{t}{2}-1$, $D_i$ is a $\rho_{i} \times a_i$ matrix for $0 \leq i \leq \frac{t}{2}-1$ for some $\{\rho_i\},\{a_i\}$,
					\item $D_0$ is a matrix with each column having weight(Hamming weight) 1 and each row having weight at least 1,
					\item  $\{A_i\}$,$\{D_i\}$,$\{B_i= \left[\frac{A_i}{D_i}\right]\}$ are  such that for $1\leq i \leq \frac{t}{2}-1$ each column of $B_i$ has weight 2, each column of $A_i$ has weight at least 1 and each row of $D_i$ has weight at least 1 and each column of $D_i$ has weight at most 1,
					\item $C$ is a matrix with each column having weight 2, $D$ is a matrix which exactly contains all the columns of $H_1$ which have weight $\geq 3$,
					\item If $J$ is the first index such that $A_J, D_J$ are empty matrices ($0 \times L$, $L \times 0$, $0 \times 0$ matrix) ($J=0$ if $D_0$ is an empty matrix) then we set $A_i,D_i$ to be empty matrices for all $J \leq i \leq \frac{t}{2}-1$ and set $a_i=0, \rho_i=0, \forall J \leq i \leq \frac{t}{2}-1$.  If none of $D_0,A_i,D_i$,$1\leq i \leq \frac{t}{2}-1$ are empty matrices, then we simply set $J = \frac{t}{2}$.
					If $J < \frac{t}{2}$, let $\alpha_1,...,\alpha_l \in \{\sum_{j=0}^{J-1}a_j+1,...,n\}$ be the indices of 2-weight columns in $H_1$ apart from the 2-weight columns corresponding to columns of $B_1,...,B_{J-1}$. We take the submatrix of $H_1$ with rows indexed by $\sum_{j=0}^{J-1}{\rho_j}+1,...,m$ and columns indexed by $\alpha_1,...,\alpha_l$ and name the submatrix as matrix $C$. Let the number of columns in $C$ be $a_{\frac{t}{2}}$. If $C$ is an empty matrix then we set $a_{\frac{t}{2}}=0$. If $J < \frac{t}{2}$, the matrix $H_1$ can be written in the form \eqref{Hform_111}
				\end{enumerate}
					\bea
					H_1 = \left[
					\scalemath{1}{
						\begin{array}{c|c|c|c|c|c|c|c|c}
							D_0 & A_1 & 0 & 0 & \hdots & 0 & 0 & 0 & \\
							\cline{1-8}
							0 & D_1 & A_2 & 0 & \hdots & 0 & 0 & 0 & \\
							\cline{1-8}
							0 & 0 & D_2 & A_3 & \hdots & 0 & 0 & 0 & \\
							\cline{1-8}
							0 & 0 & 0 & D_3 & \hdots & 0 & 0 & 0 &  \\
							\cline{1-8}
							\vdots & \vdots & \vdots & \vdots & \hdots & \vdots & \vdots & \vdots & D\\
							\cline{1-8}
							0 & 0 & 0 & 0 & \hdots & A_{J-2} & 0 & 0 & \\
							\cline{1-8}
							0 & 0 & 0 & 0 & \hdots & D_{J-2} & A_{J-1} & 0 & \\
							\cline{1-8}
							0 & 0 & 0 & 0 & \hdots & 0 & D_{J-1} & 0 & \\
							\cline{1-8}
							0 & 0 & 0 & 0 & \hdots & 0 & 0 & C & \\
							\end{array} 
							}
							\right], \ \ \label{Hform_111}
							\eea
%				where 
%				\begin{itemize}
%%					\item the rows of $H_1$ are labeled ${1,2,...,m}$ and columns are labeled ${1,2,...n}$,
%%					\item the matrix $A_i$ has dimension $\rho_{i-1} \times a_i$ for $1\leq i \leq \frac{t}{2}-1$, the matrix $D_i$ has dimension $\rho_{i} \times a_i$ for  $0 \leq i \leq \frac{t}{2}-1$,
%%					\item $D_0$ is a $\rho_0 \times a_0$ matrix with each column having Hamming weight 1 and each row having Hamming weight at least 1,
%%					\item  $\{A_i\}$,$\{D_i\}$,$\{B_i:= \left[\frac{A_i}{D_i}\right]\}$ are matrices such that each column of $B_i$ has Hamming weight 2, each column of $A_i$ has Hamming weight at least 1 and each row of $D_i$ has Hamming weight at least 1,
%%					\item $C$ is a matrix with each column having Hamming weight 2,
%%					\item $D$ is a matrix with each column having Hamming weight at least 3,	
%					\item If $J$ is the first index such that $A_J, D_J$ ($J=0$ if $D_0$ is an empty matrix) are empty matrices ($0 \times L$, $L \times 0$, $0 \times 0$ matrix) then we set $A_i,D_i$ to be empty matrices for all $\frac{t}{2}-1 \geq i \geq J$ and set $a_i=0, \rho_i=0, \frac{t}{2}-1 \geq i \geq J$ and place all of the 2-weight columns apart from the 2-weight columns corresponding to $B_1$ to $B_{J-1}$ in $C$. Let the number of columns in $C$ be $a_{\frac{t}{2}}$. If $C$ has no columns then we set $a_{\frac{t}{2}}=0$.
%				\end{itemize}
				The entire rate bound derivation is correct and all the bounds will hold with $a_i=0, \rho_i=0, J \leq i \leq \frac{t}{2}-1$. If $l$  $(<J)$ is the first index such that $D_l$ is an empty matrix with $A_l$ a non empty matrix then the proof of the following claim \ref{claim1} (since in the proof of claim \ref{claim1}, we prove the claim for $A_l$ first and then proceed to $D_l$ and since $D_0,A_j,D_j$ must be non empty $\forall j < l$) will imply that $A_l$ has column weight 1 which will imply $D_l$ cannot be empty. Hence the case $A_l$, a non empty matrix and $D_l$, an empty matrix cannot occur.
				Although we have to prove the following claim \ref{claim1} for $A_i,D_i$, $1 \leq i \leq J-1$, $D_0$, we assume all $D_0,$ $A_i,D_i$, $1 \leq i \leq \frac{t}{2}-1$ to be non-empty matrices and prove the claim. Since the proof is by induction, the induction can be made to stop at $J-1$ (induction starts at 0) and the proof is unaffected by it.
				\begin{claim}\label{claim1}
					For $1 \leq i \leq \frac{t}{2}-1$, $A_i$ is a matrix with each column having weight 1 and
					for $0 \leq j \leq \frac{t}{2}-1$, $D_j$ is a matrix with each row and each column having weight 1. 
				\end{claim}
%				\begin{claim}\label{claim1}
%					
%					$\{A_i\}$ are matrices with each column having Hamming weight 1 and
%					$\{D_i\}$ are matrices with each row and each column having Hamming weight 1. 
%					
%					
%				\end{claim}
				\begin{proof}
					We use the fact that $d_{min}(\mathcal{C}) \geq t+1$ to prove the above claim. Proof is given in Appendix \ref{AppendixA}.
				\end{proof}
				
				By Claim \ref{claim1}, after permutation of columns of $H_1$ (in \eqref{Hform} or \eqref{Hform_111} depending on $J$) within the columns labeled by the set $\{\sum_{l=0}^{j-1}a_l+1,...\sum_{l=0}^{j-1}a_l+a_{j}\}$ for $0 \leq j \leq J-1$, the matrix $D_j,0 \leq j \leq J-1$ can be assumed to be a diagonal matrix with non-zero entries along the diagonal and hence $\rho_{i}=a_{i}$, $\forall 0 \leq i \leq \frac{t}{2}-1$. 
				
				Since the sum of the weights of the columns of $A_i, 1 \leq i \leq \frac{t}{2}-1$ must equal the sum of the row weights and since each row of $A_i$ for $i \leq J-1$ can have weight at most $r$ and not $r+1$ due to weight one rows in $D_{i-1}$, and since for $\frac{t}{2}-1 \geq i \geq J$, $A_i$ is an empty matrix and we have set $a_i=0$, we obtain:\\
				\bea
				\text{For } 1 \leq i \leq \frac{t}{2}-1 : & & \notag \\
				\rho_{i-1} r & \geq & a_i, \notag \\
				a_{i-1} r & \geq & a_i. \label{Ineq1}
				\eea
				For some $p \geq 0$,\\
				\bea
				\sum_{i=0}^{\frac{t}{2}-1} a_i + p =m.  \label{Ineq2}
				\eea 
					By equating sum of row weights of $C$, with sum of column weights of $C$, we obtain:
				\bea
				2 a_{\frac{t}{2}} & \leq & (a_{\frac{t}{2}-1}+p) (r+1)-a_{\frac{t}{2}-1}. \label{Ineq3}	    
				\eea
				 Note that if $C$ is an empty matrix then also the inequality \eqref{Ineq3} is true as we would have set $a_{\frac{t}{2}}=0$. If $J < \frac{t}{2}$ and $C$ a non-empty matrix then the number of rows in $C$ is $p$ with each column of $C$ having weight 2, hence the inequality \eqref{Ineq3} is still true.\\
				Substituting \eqref{Ineq2} in \eqref{Ineq3} we get:
				\bea
				2 a_{\frac{t}{2}} & \leq & (m-\sum_{i=0}^{\frac{t}{2}-2} a_i) (r+1)-(m-\sum_{i=0}^{\frac{t}{2}-2} a_i-p), \notag \\
				2 a_{\frac{t}{2}} & \leq & (m-\sum_{i=0}^{\frac{t}{2}-2} a_i) r+p. \label{Ineq4}
				\eea
				
				By equating sum of row weights of $H_1$, with sum of column weights of $H_1$, we obtain:
				\bea
				m(r+1) & \geq & a_0 + 2(\sum_{i=1}^{\frac{t}{2}} a_i) + 3 (n-\sum_{i=0}^{\frac{t}{2}} a_i), \label{Ineq_star1} 
				\eea
				If $J < \frac{t}{2}$ then $a_i=0$, $\forall J \leq i \leq \frac{t}{2}-1$. If $C$ is an empty matrix then $a_{\frac{t}{2}}=0$. Hence the inequality \eqref{Ineq_star1} is true irrespective of whether $J=\frac{t}{2}$ or $J< \frac{t}{2}$ (even if $C$ is an empty matrix).\\
				\bea
				\text{From \eqref{Ineq_star1} :} & &  \notag \\
				m(r+1) & \geq & 3n-2a_0 -(\sum_{i=1}^{\frac{t}{2}} a_i). \label{Ineq5} 
				\eea
				
				Our basic inequalities are \eqref{Ineq1},\eqref{Ineq2},\eqref{Ineq3},\eqref{Ineq_star1}. We manipulate these 4 inequalities to derive the bound on rate.
				
				Substituting \eqref{Ineq2} in \eqref{Ineq5} we get:
				\bea
				m(r+1) & \geq & 3n-a_0 -a_{\frac{t}{2}}-(m-p), \notag \\
				m(r+2) & \geq & 3n+p-a_0 -a_{\frac{t}{2}}. \label{Ineq6} 
				\eea
				Substituting \eqref{Ineq4} in \eqref{Ineq6}, we get:
				\bea
				m(r+2) & \geq & 3n+p-a_0 - \left( \frac{(m-\sum_{i=0}^{\frac{t}{2}-2} a_i) r+p}{2} \right), \notag \\
				m(r+2+\frac{r}{2}) & \geq & 3n+\frac{p}{2}-a_0 + \left(\sum_{i=0}^{\frac{t}{2}-2} a_i \right) \frac{r}{2}. \label{Ineq7}
				\eea
				
				From \eqref{Ineq2}, for any $0 \leq j \leq \frac{t}{2}-2$:
				\bea
				a_{\frac{t}{2}-2-j} = m-\sum_{i=0}^{\frac{t}{2}-2-j-1} a_i - \sum_{i=\frac{t}{2}-2-j+1}^{\frac{t}{2}-1} a_i - p. \label{Ineq10}
				\eea
				Subtituting \eqref{Ineq1} for $\frac{t}{2}-2-j+1 \leq i \leq \frac{t}{2}-1$ in \eqref{Ineq10}, we get: 
				\bea
				a_{\frac{t}{2}-2-j} & \geq & m-\sum_{i=0}^{\frac{t}{2}-2-j-1} a_i - \sum_{i=\frac{t}{2}-2-j+1}^{\frac{t}{2}-1} a_{\frac{t}{2}-2-j} r^{i-(\frac{t}{2}-2-j)} - p, \notag \\
				a_{\frac{t}{2}-2-j} & \geq & m-\sum_{i=0}^{\frac{t}{2}-2-j-1} a_i - \sum_{i=1}^{j+1} a_{\frac{t}{2}-2-j} r^{i} - p, \notag \\
				a_{\frac{t}{2}-2-j} & \geq & \frac{m-\sum_{i=0}^{\frac{t}{2}-2-j-1} a_i-p}{ 1+ \sum_{i=1}^{j+1} r^{i} }. \label{Ineq11}
				\eea
				Let\\
				\bea
				\delta_0 & = & \frac{r}{2}, \label{Ineq12} \\
				\delta_{j+1} & = & \delta_j - \frac{\delta_j}{1+\sum_{i=1}^{j+1}r^i}. \label{Ineq13}
				\eea
				Let us prove the following inequality by induction for $0 \leq J_1 \leq \frac{t}{2}-2$,
				\bea
				m(r+2+\delta_{J_1}) & \geq & 3n + p\left( \frac{1}{2}+\delta_{J_1} - \frac{r}{2}\right)-a_0  + \left(\sum_{i=0}^{\frac{t}{2}-2-J_1} a_i \right) \delta_{J_1}. \label{Ineq9}
				\eea
				\eqref{Ineq9} is true for $J_1=0$ by \eqref{Ineq7}. Hence \eqref{Ineq9} is proved for $t=4$ and the range of $J_1$ is vacuous for $t=2$. Hence assume $t>4$. Hence let us assume \eqref{Ineq9} is true for $ J_1$ such that $ \frac{t}{2}-3 \geq J_1 \geq 0$ and prove it for $J_1+1$. 
				Substituting \eqref{Ineq11} for $j=J_1$ in \eqref{Ineq9}, we get:
				\bea
				m(r+2+\delta_{J_1}) & \geq & 3n + p\left( \frac{1}{2}+\delta_{J_1} - \frac{r}{2}\right)-a_0  + \left(\sum_{i=0}^{\frac{t}{2}-2-J_1-1} a_i \right) \delta_{J_1} + \left( \frac{m-\sum_{i=0}^{\frac{t}{2}-2-J_1-1} a_i-p}{ 1+ \sum_{i=1}^{J_1+1} r^{i} } \right) \delta_{J_1}, \notag \\
				m(r+2+\delta_{J_1}-\frac{\delta_{J_1}}{1+ \sum_{i=1}^{J_1+1} r^{i} }) & \geq & 
				3n + p\left( \frac{1}{2}+\delta_{J_1} -\frac{\delta_{J_1}}{1+ \sum_{i=1}^{J_1+1} r^{i} }- \frac{r}{2}\right)-a_0  + \notag \\ 
				& & \left(\sum_{i=0}^{\frac{t}{2}-2-J_1-1} a_i \right) \left(\delta_{J_1} - \frac{\delta_{J_1}}{1+ \sum_{i=1}^{J_1+1} r^{i} }\right). \label{Ineq141}
				\eea
				Substituing \eqref{Ineq13} in \eqref{Ineq141}, we obtain
				\bea
				m(r+2+\delta_{J_1+1}) & \geq & 
				3n + p\left( \frac{1}{2}+\delta_{J_1+1}- \frac{r}{2}\right)-a_0  + \left(\sum_{i=0}^{\frac{t}{2}-2-J_1-1} a_i \right) \delta_{J_1+1}. \label{Ineq14}
				\eea
				Hence \eqref{Ineq9} is proved  for any $0 \leq J_1 \leq \frac{t}{2}-2$ for $t \geq 4$. Hence writing \eqref{Ineq9} for $J_1 = \frac{t}{2}-2$ for $t\geq 4$, we obtain:
				\bea
				m(r+2+\delta_{\frac{t}{2}-2}) & \geq & 3n + p\left( \frac{1}{2}+\delta_{\frac{t}{2}-2} - \frac{r}{2}\right)-a_0  + (a_0) \delta_{\frac{t}{2}-2}, \notag \\
				m(r+2+\delta_{\frac{t}{2}-2}) & \geq & 3n + p\left( \frac{1}{2}+\delta_{\frac{t}{2}-2} - \frac{r}{2}\right)+ a_0 (\delta_{\frac{t}{2}-2}-1). \label{Ineq15} 
				\eea
				It can be seen that $\delta_j$ for $r \geq 2$ has a product form as:
				\bea
				\delta_j = \frac{r}{2}\left(\frac{r^{j+1}-r^{j}}{r^{j+1}-1}\right). \label{Ineq16}
				\eea
				Hence for $r \geq 3$, $t \geq 4$:
				\bean
				\delta_{\frac{t}{2}-2}=\frac{r}{2}\left(\frac{r^{\frac{t}{2}-1}-r^{\frac{t}{2}-2}}{r^{\frac{t}{2}-1}-1}\right) \geq 1.
				\eean
				Hence we can substitute \eqref{Ineq11} for $j=\frac{t}{2}-2$ in \eqref{Ineq15} :
				\bea
				m(r+2+\delta_{\frac{t}{2}-2}) & \geq & 3n + p\left( \frac{1}{2}+\delta_{\frac{t}{2}-2} - \frac{r}{2}\right)+ \left( \frac{m-p}{1+\sum_{i=1}^{\frac{t}{2}-1}r^i} \right) (\delta_{\frac{t}{2}-2}-1), \notag \\
				m(r+2+\delta_{\frac{t}{2}-2}-\frac{\delta_{\frac{t}{2}-2}}{1+\sum_{i=1}^{\frac{t}{2}-1}r^i}+\frac{1}{1+\sum_{i=1}^{\frac{t}{2}-1}r^i}) & \geq & 3n + p\left( \frac{1}{2}+\delta_{\frac{t}{2}-2}-\frac{\delta_{\frac{t}{2}-2}}{1+\sum_{i=1}^{\frac{t}{2}-1}r^i} + \frac{1}{1+\sum_{i=1}^{\frac{t}{2}-1}r^i}- \frac{r}{2}\right). \label{Ineq17}
				\eea
				Substituting \eqref{Ineq13} in \eqref{Ineq17}, we obtain:
				\bea
				m\left(r+2+\delta_{\frac{t}{2}-1}+ \frac{1}{1+\sum_{i=1}^{\frac{t}{2}-1}r^i}\right) & \geq & 3n + p\left( \frac{1}{2}+\delta_{\frac{t}{2}-1}+ \frac{1}{1+\sum_{i=1}^{\frac{t}{2}-1}r^i}- \frac{r}{2}\right). \label{Ineq18}
				\eea
				Using \eqref{Ineq16}, we obtain:
				\bean
				\left( \frac{1}{2}+\delta_{\frac{t}{2}-1}+ \frac{1}{1+\sum_{i=1}^{\frac{t}{2}-1}r^i}- \frac{r}{2}\right) \geq 0.
				\eean
				Hence \eqref{Ineq18} implies:
				\bea
				m\left(r+2+\delta_{\frac{t}{2}-1}+ \frac{1}{1+\sum_{i=1}^{\frac{t}{2}-1}r^i}\right) & \geq & 3n.  \label{Ineq19}
				\eea
				\eqref{Ineq19} after some algebraic manipulations gives the required bound on $1-\frac{m}{n}$ and hence on $\frac{k}{n}$ as stated in the theorem. Note that although the derivation is valid for $r\geq 3$, $t \geq 4$, the final bound given in the theorem is correct and tight for $t =2$. The bound for $t=2$ can be derived specifically by substituting $a_0 \leq m$ in \eqref{Ineq7} and noting that $p \geq 0$.
			\end{proof}
			An alternative proof for Theorem \ref{rate_teven} is given below by using linear programming:
			\begin{proof}
				The inequalities \eqref{Ineq1},\eqref{Ineq3} and \eqref{Ineq5} are linear inequalities and are written in matrix form as:\\
				\begin{equation*}
				A\mathbf{\underline{x}} \geq \mathbf{\underline{b}}
				\end{equation*}
				where
				\bea
				A & = & \left[
				\begin{array}{c c c c c c c c}
					r & -1 & 0 & \hdots & 0 & 0 & 0 & 0 \\
					0 & r & -1 & \hdots & 0 & 0 & 0 & 0 \\
					\vdots & \vdots & \vdots & \ddots & \vdots & \vdots & \vdots & \vdots \\
					0 & 0 & 0 & \hdots & r & -1 & 0 & 0 \\
					0 & 0 & 0 & \hdots & 0 & r & -2 & (r+1) \\
					(r+3) & (r+2) & (r+2) & \hdots & (r+2) & (r+2) & 1 & (r+1)
				\end{array} \right] \label{Amat}
				\eea
				which is a $(\frac{t}{2} + 1) \times (\frac{t}{2} + 2) $ matrix and
				\bea 
				\mathbf{\underline{x}} = \left[ \begin{array}{c c c c c}
					 a_0 & a_1 & \hdots & a_{\frac{t}{2}} & p
				\end{array} \right]^T, \mathbf{\underline{b}} = \left[ \begin{array}{c c c c c}
					0 &	0 &	\hdots & 0 & 3n
				\end{array}\right]^T
				\eea
				where $\mathbf{\underline{x}}$ is a $(\frac{t}{2} + 2) \times 1$ matrix and $\mathbf{\underline{b}}$ is a $(\frac{t}{2} + 1) \times 1$ matrix
				The problem of finding an upper bound on rate of the code now becomes one of minimizing $m = \mathbf{\underline{c}}^T\mathbf{\underline{x}}$, which is a linear objective function where $\mathbf{\underline{c}} = \left[ \begin{array}{c c c c c c}
				1 & 1 & \hdots & 1 & 0 & 1
				\end{array}\right]^T$ is a $(\frac{t}{2} + 2) \times 1$ matrix. Also by definition of $\mathbf{\underline{x}}$, $\mathbf{\underline{x}} \geq 0$. This is now in a standard form of a linear program formulation as:
				\begin{align*}
				\text{minimize } & \mathbf{\underline{c}}^T\mathbf{\underline{x}}\\
				\text{s.t. } A\mathbf{\underline{x}} & \geq \mathbf{\underline{b}}, \mathbf{\underline{x}} \geq 0
				\end{align*}
				The dual problem of the above is
				\begin{align*}
				\text{maximize } & \mathbf{\underline{b}}^T\mathbf{\underline{\lambda}}\\
				\text{s.t. } A^T\mathbf{\underline{\lambda}} & \leq \mathbf{\underline{c}}, \mathbf{\underline{\lambda}} \geq 0
				\end{align*}
				We will solve the dual problem by writing it in standard $\text{minimize}$ $-\mathbf{\underline{b}}^T\underline{\lambda}$ form and using the simplex method. Let us introduce slack variables $s_1$,...,$s_{\frac{t}{2}+2}$ and re-write the constraints as
				\begin{align*}
				B\mathbf{\underline{v}} = \mathbf{\underline{c}}, \text{ }\mathbf{\underline{v}} \geq 0,
				\end{align*} 
				where
				\bea
				B & = & \left[
				\begin{array}{c c c c c c c c c c c c c}
					r & 0 & 0 & \hdots & 0 & (r+3) & 0 & 1 & 0 & 0 & \hdots & 0 & 0 \\
					-1 & r & 0 & \hdots & 0 & (r+2) & 0 & 0 & 1 & 0 &\hdots & 0 & 0\\
					0 & -1 & r & \hdots & 0 & (r+2) & 0 & 0 & 0 & 1 &\hdots & 0 & 0\\
					\vdots & \vdots & \ddots & \ddots & \vdots & \vdots & \vdots & \vdots & \vdots & \vdots & \ddots & \vdots & \vdots\\
					0 & 0 & \hdots & -1 & r & (r+2) & 0 & 0 & 0 & 0 & \hdots & 1 & 0 \\
					0 & 0 & 0 & \hdots & -2 & 1 & 0 & 0 & 0 & 0 & \hdots & 0 & 1 \\
					0 & 0 & 0 & \hdots & (r+1) & (r+1) & 1 & 0 & 0 & 0 & \hdots & 0 & 0 \\ 
				\end{array} \right] \label{Amat},
				\eea
				and 
				\bea 
				\mathbf{\underline{v}} = \left[ \begin{array}{c c c c c c c}
					\lambda_1 & \lambda_2 & \hdots & \lambda_{\frac{t}{2}+1} & s_1 & \hdots & s_{\frac{t}{2}+2}
				\end{array} \right]^T.
			\eea
			With this, the objective function now is $\mathbf{\underline{d}}^T\mathbf{\underline{v}}$, where $\mathbf{\underline{d}} = \left[\begin{array}{c c c c c}
			-\mathbf{\underline{b}}^T & 0 & 0 & \hdots & 0
			\end{array} \right]^T$ which is a $(t+3) \times 1$ matrix.
		 Define $\mathbf{\underline{v}}_{BV}=[\lambda_1,...,\lambda_{\frac{t}{2}+1},s_1]^T$.
			
			We pick the variables $\beta_1=\lambda_1$,...,$\beta_{\frac{t}{2}+1}=\lambda_{\frac{t}{2}+1}$, $\beta_{\frac{t}{2}+2}=s_1$ as ``basic variables" and the rest, called ``non-basic variables" $\alpha_1=s_2,...,\alpha_{\frac{t}{2}+1}=s_{\frac{t}{2}+2}$ will be set to $0$. The set of basic variables is chosen such that the matrix formed by the columns of $B$ corresponding to these basic variables is a full-rank square matrix $B_{BV}$. The remaining columns of $B$ will give a matrix $B_{NBV}$. The system of equations is now in the following form:
			\begin{align*}
			\left[\begin{array}{c c}
			B_{BV} & B_{NBV}
			\end{array}\right]
			\left[\begin{array}{c}
			\mathbf{\underline{v}}_{BV} \\
			\hline
			\mathbf{\underline{0}}
			\end{array}
			\right] = \mathbf{\underline{c}}
			\end{align*}
			Therefore we will equivalently solve
			\begin{align*}
			B_{BV}\mathbf{\underline{v}}_{BV} = \mathbf{\underline{c}}
			\end{align*}
			The above system of equations can be solved in closed form to get the following:
			\begin{align}
			\lambda_{\frac{t}{2}+1} & = \frac{2\sum_{i = 0}^{\frac{t}{2}-1}r^i}{3(r^{\frac{t}{2}}+2\sum_{i = 0}^{\frac{t}{2}-1}r^i)}, \label{LP1} \\
			s_1 & = \frac{1}{r^{\frac{t}{2}}+2\sum_{i = 0}^{\frac{t}{2}-1}r^i}, \label{LP2}\\
			\lambda_{j+1} & = \frac{r^{\frac{t}{2}} - 3r^{\frac{t}{2}-(j+1)} + 2}{3(r-1)(r^{\frac{t}{2}}+2\sum_{i = 0}^{\frac{t}{2}-1}r^i)},\text{   for }0 \leq j \leq \frac{t}{2}-1 \label{LP3}
			\end{align}
			which are non-negative if $r \geq 3$. Hence the solution given by \eqref{LP1},\eqref{LP2},\eqref{LP3} is a basic feasible solution. Let the elements of the vector $\mathbf{\underline{d}}$ be indexed by the elements of the vector $\mathbf{\underline{v}}$ i.e., $i^{th}$ component of $\mathbf{\underline{d}}$ is indexed by $i^{th}$ component of $\mathbf{\underline{v}}$. To check for optimality we check if the ``reduced cost coefficients" $r_{\alpha_i} = d_{\alpha_i} - z_{\alpha_i}$ are non-negative, for every non-basic variable $\alpha_i,1\leq i \leq \frac{t}{2}+1$. We note that for the above made choice of basic and non-basic variables, in the vector $\mathbf{\underline{d}}$ only $d_{\beta_{\frac{t}{2}+1}} = -3n$ is non-zero. The quantity $z_{\alpha_i}$ is defined as follows:
			\begin{equation*}
			z_{\alpha_i} = \sum_{j = 1}^{{\frac{t}{2}+2}}d_{\beta_j}y_{(j,\alpha_i)} = d_{\beta_{\frac{t}{2}+1}}y_{(\frac{t}{2}+1,\alpha_i)} = -3ny_{(\frac{t}{2}+1,\alpha_i)}
			\end{equation*}
			where $y_{(\frac{t}{2}+1,\alpha_i)}$ are as shown in the row reduced echelon form of matrix B below:
			\begin{align*}
			B_{rref} = \left[\begin{array}{c c c c c c c c c c c c c}
			1 & 0 & 0 & \hdots & 0 & 0 & 0 & y_{(1,\alpha_1)} & y_{(1,\alpha_2)} & y_{(1,\alpha_3)} & \hdots & y_{(1,\alpha_{\frac{t}{2}})} & y_{(1,\alpha_{\frac{t}{2}+1})} \\
			0 & 1 & 0 & \hdots & 0 & 0 & 0 & y_{(2,\alpha_1)} & y_{(2,\alpha_2)} & y_{(2,\alpha_3)} & \hdots & y_{(2,\alpha_{\frac{t}{2}})} & y_{(2,\alpha_{\frac{t}{2}+1})}\\
			0 & 0 & 1 & \hdots & 0 & 0 & 0 & y_{(3,\alpha_1)} & y_{(3,\alpha_2)} & y_{(3,\alpha_3)} & \hdots & y_{(3,\alpha_{\frac{t}{2}})} & y_{(3,\alpha_{\frac{t}{2}+1})}\\
			\vdots & \vdots & \vdots & \ddots & \vdots & \vdots & \vdots & \vdots & \vdots & \vdots & \ddots & \vdots & \vdots\\
			0 & 0 & 0 & \hdots & 1 & 0 & 0 & y_{(\frac{t}{2},\alpha_1)} & y_{(\frac{t}{2},\alpha_2)} & y_{(\frac{t}{2},\alpha_3)} & \hdots & y_{(\frac{t}{2},\alpha_{\frac{t}{2}})} & y_{(\frac{t}{2},\alpha_{\frac{t}{2}+1})}\\
			0 & 0 & 0 & \hdots & 0 & 1 & 0 & y_{(\frac{t}{2}+1,\alpha_1)} & y_{(\frac{t}{2}+1,\alpha_2)} & y_{(\frac{t}{2}+1,\alpha_3)} & \hdots & y_{(\frac{t}{2}+1,\alpha_{\frac{t}{2}})} & y_{(\frac{t}{2}+1,\alpha_{\frac{t}{2}+1})}\\
			0 & 0 & 0 & \hdots & 0 & 0 & 1 & y_{(\frac{t}{2}+2,\alpha_1)} & y_{(\frac{t}{2}+2,\alpha_2)} & y_{(\frac{t}{2}+2,\alpha_3)} & \hdots & y_{(\frac{t}{2}+2,\alpha_{\frac{t}{2}})} & y_{(\frac{t}{2}+2,\alpha_{\frac{t}{2}+1})}\\ 
			\end{array}
			\right]
			\end{align*}
			We observe that in $B$ to $B_{rref}$, to row-$(\frac{t}{2}+1)$ only non-negative linear combinations of the rows above it are added, entries of which are either $0$ or $1$. Therefore $r_{\alpha_i} \geq 0$ for $\alpha_i$ all non-basic variables. Hence $y_{(\frac{t}{2}+1,\alpha_1)},...,y_{(\frac{t}{2}+1,\alpha_{\frac{t}{2}+1})} \geq 0$. Hence the basic solution given by \eqref{LP1}, \eqref{LP2} and \eqref{LP3} is an ``optimal basic feasible" solution.\\
			By the theorem of strong duality the optimal solutions of the primal problem and the dual problem are equal. Therefore the minimum value of $m$ is $3n\lambda_{\frac{t}{2}+1} = \frac{n2\sum_{i = 0}^{\frac{t}{2}-1}r^i}{(r^{\frac{t}{2}}+2\sum_{i = 0}^{\frac{t}{2}-1}r^i)}$.\\
			Hence we get the upper bound on the rate:
			\begin{equation*}
			\frac{k}{n} \leq 1 - \frac{m}{n} \leq \frac{r^{\frac{t}{2}}}{r^{\frac{t}{2}} + 2 \sum\limits_{i=0}^{\frac{t}{2}-1} r^i}
			\end{equation*}
			We now pick a solution for the primal problem and show that it is feasible and gives the optimal objective function value.
			\begin{align*}
			a_i & = \frac{2nr^i}{r^{\frac{t}{2}}+2\sum_{i = 0}^{\frac{t}{2}-1}r^i},\text{ for }0 \leq i \leq \frac{t}{2}-1\\
			a_{\frac{t}{2}} & = \frac{nr^{\frac{t}{2}}}{r^{\frac{t}{2}}+2\sum_{i = 0}^{\frac{t}{2}-1}r^i},\text{   }
			p = 0.
			\end{align*}
			It is easy to check that this solution satisfies the first $\frac{t}{2}$ constraints of the primal problem with equality. It remains to check the following:
			\begin{equation*}
			(r+3)a_0+(r+2)\sum_{i = 1}^{\frac{t}{2}-1}a_i+a_{\frac{t}{2}} \geq 3n.
			\end{equation*}
			Upon simplification, it is seen that the above inequality is met with equality. Therefore the chosen solution is a feasible solution. It is also easy to check that the solution gives the optimal value of the objective function. Hence it is an optimal feasible solution. We thus conclude that a code having the above chosen values will have the optimal rate.
			\end{proof} 
			\subsection{Upper Bound on Rate of an $(n,k,r,t)_{seq}$ Code for $t$ Odd:}
			
			In this subsection we provide an upper bound on the rate of an $(n,k,r,t)_{seq}$ code whenever $t$ is a positive odd integer and $r \geq 3$. 
			
			\begin{thm}\label{thm_todd}
				Let $\mathcal{C}$ be an $(n,k,r,t)_{seq}$ code over a field $\mathbb{F}_q$. Let $t=2s-1, s \geq 1$ and $r \geq 3$. Then 
				
				\begin{equation}
				\frac{k}{n} \leq \frac{r^{s}}{r^{s} + 2 \sum\limits_{i=1}^{s-1} r^i + 1} \label{Thm2}.
				\end{equation}
				
			\end{thm}

			\begin{proof}
				As earlier, we again begin by setting $\mathcal{B}_0 = \text{span}({\underline{c} \in \mathcal{C}^{\perp} : w_H(\underline{c}) \leq r+1})$.
				Let $m$ be the dimension of $\mathcal{B}_0$. Let ${\underline{c}_1,...\underline{c}_m}$ (row vectors) be a basis of $\mathcal{B}_0$ such that $w_H(\underline{c}_i) \leq r+1$.
				\bean
				\text{Let } H_1  & = &  \left[ 
				\begin{array}{c}
					\underline{c}_1 \\
					\vdots \\
					\underline{c}_m			
				\end{array}\right].
				\eean
				It follows that $H_1$ is a parity check matrix of an $(n,n-m,r,t)_{seq}$ code as its row space contains all the codewords of Hamming weight at most $r+1$ which are contained in $\mathcal{C}^\perp$. Also, 
				\bean
				\frac{k}{n} \leq 1 - \frac{m}{n}.
				\eean
%				\begin{equation}
%				\text{Let } \mathcal{B}_0 = \text{span}({c \in \mathcal{C}^{\perp} : w_H(c) \leq r+1}).\nonumber
%				\end{equation}
%				Let $m$ be the dimension of $\mathcal{B}_0$. Let ${c_1,...c_m}$ be a basis of $\mathcal{B}_0$ such that $w_H(c_i) \leq r+1$.
%				\bean
%				\text{Let } H_1  & = &  \left[ 
%				\begin{array}{c}
%					c_1 \\
%					\vdots \\
%					c_m			
%				\end{array}\right].
%				\eean
%				Now from the definition of an $(n,k,r,t)_{seq}$code, it can be easily seen that $H_1$ is a parity check matrix of an $(n,k,r,t)_{seq}$code as its row span contains all the codewords of Hamming weight at most $r+1$ which are in $\mathcal{C}^\perp$. Also, 
%				\bean
%				\frac{k}{n} \leq 1 - \frac{m}{n}.
%				\eean
				We now proceed to derive an upper bound on $1 - \frac{m}{n}$.
				Again it can be seen that the matrix $H_1$, after permutation of rows and columns can be written in the following form
				\bea
				H_1 & = & \left[
				\begin{array}{c|c|c|c|c|c|c|c|c}
					D_0 & A_1 & 0 & 0 & \hdots & 0 & 0 & 0 & \\
					\cline{1-8}
					0 & D_1 & A_2 & 0 & \hdots & 0 & 0 & 0 & \\
					\cline{1-8}
					0 & 0 & D_2 & A_3 & \hdots & 0 & 0 & 0 & \\
					\cline{1-8}
					0 & 0 & 0 & D_3 & \hdots & 0 & 0 & 0 &  \\
					\cline{1-8}
					\vdots & \vdots & \vdots & \vdots & \ddots & \vdots & \vdots & \vdots & D\\
					\cline{1-8}
					0 & 0 & 0 & 0 & \hdots & A_{s-2} & 0 & 0 & \\
					\cline{1-8}
					0 & 0 & 0 & 0 & \hdots & D_{s-2} & A_{s-1} & 0 & \\
					\cline{1-8}
					0 & 0 & 0 & 0 & \hdots & 0 & D_{s-1} &  & \\
					\cline{1-7}
					0 & 0 & 0 & 0 & \hdots & 0 & 0 & C & \\
				\end{array} \right] \label{Hform2},
				\eea
%				where 
%				\begin{enumerate}
%					\item the rows of $H_1$ are labeled ${1,2,...,m}$ and columns are labeled ${1,2,...n}$,
%					\item $A_i$ is a $\rho_{i-1} \times a_i$ matrix for $1\leq i \leq \frac{t}{2}-1$, $D_i$ is a $\rho_{i} \times a_i$ matrix for $0 \leq i \leq \frac{t}{2}-1$ for some $\{\rho_i\},\{a_i\}$,
%					\item $D_0$ is a matrix with each column having weight(Hamming weight) 1 and each row having weight at least 1,
%					\item  $\{A_i\}$,$\{D_i\}$,$\{B_i= \left[\frac{A_i}{D_i}\right]\}$ are  such that for $1\leq i \leq \frac{t}{2}-1$ each column of $B_i$ has weight 2, each column of $A_i$ has weight at least 1 and each row of $D_i$ has weight at least 1 and each column of $D_i$ has weight at most 1,
%					\item $C$ is a matrix with each column having weight 2, $D$ is a matrix which exactly contains all the columns of $H_1$ which have weight $\geq 3$,
%					\item If $J$ is the first index such that $A_J, D_J$ ($J=0$ if $D_0$ is an empty matrix) are empty matrices ($0 \times L$, $L \times 0$, $0 \times 0$ matrix) then we set $A_i,D_i$ to be empty matrices for all $J \leq i \leq \frac{t}{2}-1$ and set $a_i=0, \rho_i=0, J \leq i \leq \frac{t}{2}-1$ and place all of the 2-weight columns apart from the 2-weight columns corresponding to $B_1$ to $B_{J-1}$ in $C$. Let the number of columns in $C$ be $a_{\frac{t}{2}}$. If $C$ has no columns then we set $a_{\frac{t}{2}}=0$.
%				\end{enumerate}
				where 
				\begin{enumerate}
					\item the rows of $H_1$ are labeled ${1,2,...,m}$ and columns are labeled ${1,2,...n}$,
					\item $A_i$ is a $\rho_{i-1} \times a_i$ matrix for  $1\leq i \leq s-1$, $D_i$ is a $\rho_{i} \times a_i$ matrix for $0 \leq i \leq s-1$ for some $\{\rho_i\},\{a_i\}$,
					\item $D_0$ is a matrix with each column having weight 1 and each row having weight at least 1,
					\item  $\{A_i\}$,$\{D_i\}$,$\{B_i= \left[\frac{A_i}{D_i}\right]\}$ are  such that for $1\leq i \leq s-1$ each column of $B_i$ has weight 2, each column of $A_i$ has weight at least 1, each row of $D_i$ has weight at least 1 and each column of $D_i$ has weight at most 1,
%					\item  $\{A_i, 1\leq i \leq s-1\}$,$\{D_i,1 \leq i \leq s-1\}$,$\{B_i:= \left[\frac{A_i}{D_i}\right]\}$ are matrices such that each column of $B_i$ has Hamming weight 2, each column of $A_i$ has Hamming weight at least 1 and each row of $D_i$ has Hamming weight at least 1,
					\item $C$ is a matrix with each column having weight 2, $D$ is a matrix which exactly contains all the columns of $H_1$ which have weight at least 3,
					\item If $J$ is the first index such that $A_J, D_J$ are empty matrices (0 x $L$ , $L$ x 0, 0 x 0 matrix) ($J=0$ if $D_0$ is an empty matrix)  then we set $A_i,D_i$ to be empty matrices for all $s-1 \geq i \geq J$ and set $a_i=0, \rho_i=0, \forall s-1 \geq i \geq J$. If none of $D_0,A_i,D_i$,$1\leq i \leq s-1$ are empty matrices, then we simply set $J = s$.
					If $J < s$, let $\alpha_1,...,\alpha_l \in \{\sum_{j=0}^{J-1}a_j+1,...,n\}$ be the indices of 2-weight columns in $H_1$ apart from the 2-weight columns corresponding to columns of $B_1,...,B_{J-1}$. We take the submatrix of $H_1$ with rows indexed by $\sum_{j=0}^{J-1}{\rho_j}+1,...,m$ and columns indexed by $\alpha_1,...,\alpha_l$ and name the submatrix as matrix $C$. Let the number of columns in $C$ be $a_s$. If $C$ is an empty matrix then we set $a_{s}=0$. If $J < s $, the matrix $H_1$ can be written in the form \eqref{Hform_11}.
				\end{enumerate}
				\bea
				H_1 = \left[
				\scalemath{1}{
					\begin{array}{c|c|c|c|c|c|c|c|c}
						D_0 & A_1 & 0 & 0 & \hdots & 0 & 0 & 0 & \\
						\cline{1-8}
						0 & D_1 & A_2 & 0 & \hdots & 0 & 0 & 0 & \\
						\cline{1-8}
						0 & 0 & D_2 & A_3 & \hdots & 0 & 0 & 0 & \\
						\cline{1-8}
						0 & 0 & 0 & D_3 & \hdots & 0 & 0 & 0 &  \\
						\cline{1-8}
						\vdots & \vdots & \vdots & \vdots & \hdots & \vdots & \vdots & \vdots & D\\
						\cline{1-8}
						0 & 0 & 0 & 0 & \hdots & A_{J-2} & 0 & 0 & \\
						\cline{1-8}
						0 & 0 & 0 & 0 & \hdots & D_{J-2} & A_{J-1} & 0 & \\
						\cline{1-8}
						0 & 0 & 0 & 0 & \hdots & 0 & D_{J-1} & 0 & \\
						\cline{1-8}
						0 & 0 & 0 & 0 & \hdots & 0 & 0 & C & \\
					\end{array} 
				}
				\right], \ \ \label{Hform_11}
				\eea
				The entire rate bound derivation is correct and all the bounds will hold with $a_i=0, \rho_i=0, J \leq i \leq s-1$. If $l$  $(<J)$ is the first index such that $D_l$ is an empty matrix with $A_l$ a non empty matrix then the proof of the following claim \ref{claim2} (since in the proof of claim \ref{claim2}, we prove the claim for $A_l$ first and then proceed to $D_l$ and since $D_0,A_j,D_j$ must be non empty $\forall j < l$) will imply that $A_l$ has column weight 1 which will imply $D_l$ cannot be empty. Hence the case $A_l$, a non empty matrix and $D_l$, an empty matrix cannot occur.
				Although we have to prove the following claim \ref{claim2} for $A_i,D_i$, $1 \leq i \leq J-1$, $D_0$, we assume all $D_0,$ $A_i,D_i$, $1 \leq i \leq s-1$ to be non-empty matrices and prove the claim. Since the proof is by induction, the induction can be made to stop at $J-1$ (induction starts at 0) and the proof is unaffected by it.
%				\\\\
%				The following claim is vacuously true for empty matrices and the entire derivation is true with $a_i=0, \rho_i=0, s-1 \geq i \geq J$.
%				Although we have to prove the following claim for $A_i,D_i$, $1 \leq i \leq J-1$, $D_0$, we assume that all $D_0,$$A_i,D_i$,$1 \leq i \leq s-1$ to be non-empty matrices and prove the claim. Since the proof is by induction, the induction can be made to stop at $J-1$ and the proof is unaffected by it.
%				\begin{claim}\label{claim1}
%					For $1 \leq i \leq \frac{t}{2}-1$, $A_i$ is a matrix with each column having weight 1 and
%					for $0 \leq j \leq \frac{t}{2}-1$, $D_j$ is a matrix with each row and each column having weight 1. 
%				\end{claim}
				\begin{claim}\label{claim2}
					
					For $1\leq i \leq s-1$, $A_i$ is a matrix with each column having weight 1 and
					for $0 \leq i \leq s-2$, $D_i$ is a matrix with each row and each column having weight 1. $D_{s-1}$ is a matrix with each column having weight 1.

				\end{claim}
				\begin{proof}
					Proof is exactly similar to the proof of claim \ref{claim1}. So we skip the proof.
				\end{proof}
				
				By Claim \ref{claim2}, after permutation of columns of $H_1$ (in \eqref{Hform2} or \eqref{Hform_11} depending on $J$)  within the columns labeled by the set $\{\sum_{l=0}^{j-1}a_l+1,...\sum_{l=0}^{j-1}a_l+a_{j}\}$ for $0 \leq j \leq min(J-1,s-2)$, the matrix $D_j,0 \leq j \leq min(J-1,s-2)$ can be assumed to be a diagonal matrix with non-zero entries along the diagonal and hence $\rho_{i}=a_{i}$, for $0 \leq i \leq s-2$. 
				
					Since the sum of the weights of the columns of $A_i, 1 \leq i \leq s-1$ must equal the sum of the row weights and since each row of $A_i$ for $i \leq J-1$ can have weight at most $r$ and not $r+1$ due to weight one rows in $D_{i-1}$, and since for $s-1 \geq i \geq J$, $A_i$ is an empty matrix and we have set $a_i=0$, we obtain:\\
					\bea
					\text{For } 1 \leq i \leq s-1 : & & \notag \\
				\rho_{i-1} r & \geq & a_i, \notag \\
				a_{i-1} r & \geq & a_i. \label{Ineq1O}
				\eea
				For some $p \geq 0$,\\
				\bea
				\rho_{s-1}+\sum_{i=0}^{s-2} a_i + p =m.  \label{Ineq2O}
				\eea 
				By equating sum of row weights of $[\frac{D_{s-1}}{0} | C]$, with sum of column weights of $[\frac{D_{s-1}}{0} | C]$, we obtain:
				\bea
				2 a_{s}+a_{s-1} & \leq & (\rho_{s-1}+p) (r+1). \label{Ineq3O}	    
				\eea
				 If $J < s$ then the number of rows in $C$ is $p$ with each column of $C$ having weight 2 and $a_{s-1}=\rho_{s-1}=0$ (and $a_s=0$ if $C$ is also an empty matrix), hence the inequality \eqref{Ineq3O} is true. If $J=s$ and $C$ is an empty matrix then also the inequality \eqref{Ineq3O} is true as we would have set $a_{s}=0$.
				
				Substituting \eqref{Ineq2O} in \eqref{Ineq3O}:
				\bea
				2 a_{s} & \leq & (m-\sum_{i=0}^{s-2} a_i) (r+1)-a_{s-1}. \label{Ineq4O}
				\eea
				
				By equating sum of row weights of $D_{s-1}$, with sum of column weights of $D_{s-1}$, we obtain (Note that if $D_{s-1}$ is an empty matrix then also the following inequality is true as we would have set $a_{s-1}=0$):
				\bea
				a_{s-1} \leq \rho_{s-1}(r+1). \label{OddSpe}
				\eea
				
			By equating sum of row weights of $H_1$, with sum of column weights of $H_1$, we obtain
				\bea
				m(r+1) & \geq & a_0 + 2(\sum_{i=1}^{s} a_i) + 3 (n-\sum_{i=0}^{s} a_i), \label{Ineq_star2}
				\eea 
					If $J < s$ then $a_i=0$ $\forall J \leq i \leq s-1$. If $C$ is an empty matrix then $a_s=0$. Hence the inequality \eqref{Ineq_star2} is true irrespective of whether $J=s$ or $J<s$ (even if $C$ is an empty matrix).\\
				\bea
				m(r+1) & \geq & 3n-2a_0 -(\sum_{i=1}^{s} a_i). \label{Ineq5O} 
				\eea
				
				Our basic inequalities are \eqref{Ineq1O},\eqref{Ineq2O},\eqref{Ineq3O},\eqref{OddSpe},\eqref{Ineq_star2}. We manipulate these 5 inequalities to derive the bound on rate.
				
				Substituting \eqref{Ineq4O} in \eqref{Ineq5O}:
				\bea
				m(r+1) & \geq & 3n-2a_0 -(\sum_{i=1}^{s-1} a_i) - \left(\frac{(m-\sum_{i=0}^{s-2} a_i) (r+1)-a_{s-1}}{2} \right). \label{Ineq6O}
				\eea
				For $s=1$, \eqref{Ineq6O} becomes:
				\bea
				m(r+1) & \geq & 3n-2a_0 - \left(\frac{m (r+1)-a_{0}}{2} \right), \notag \\        m\frac{3(r+1)}{2} & \geq & 3n-\frac{3}{2}a_0. \label{Ineq7O} 
				\eea
				Substituting  \eqref{OddSpe} in \eqref{Ineq7O}:
				\bea
				m\frac{3(r+1)}{2} & \geq & 3n-\frac{3}{2}\rho_0 (r+1), \label{Ineq8O} 
				\eea
				Substituting  \eqref{Ineq2O} in \eqref{Ineq8O}:
				\bea
				m\frac{3(r+1)}{2} & \geq & 3n-\frac{3}{2}(m-p) (r+1), \notag \\
				\text{Since $p \geq 0$, } 3m(r+1) & \geq & 3n. \label{Ineq9O}
				\eea
				\eqref{Ineq9O} implies,
				\bea
				\frac{k}{n} \leq \frac{r}{r+1}. \label{Ineq10O}
				\eea
				\eqref{Ineq10O} proves the bound \eqref{Thm2} for $s=1$. Hence from now on we assume $s \geq 2$.\\
				For $s \geq 2$, \eqref{Ineq6O} implies:
				\bea
				m \frac{3(r+1)}{2} & \geq & 3n+ a_0 \left( \frac{r+1}{2}-2 \right) +(\sum_{i=1}^{s-2} a_i) \left( \frac{r+1}{2}-1 \right) - \frac{a_{s-1}}{2}. \label{Ineq11O} 
				\eea
				Substituting  \eqref{Ineq1O} in \eqref{Ineq11O} and since $r \geq 3$:    
				\bea
				m \frac{3(r+1)}{2} & \geq & 3n+ \frac{a_{s-1}}{r^{s-1}} \left( \frac{r+1}{2}-2 \right) +(\sum_{i=1}^{s-2} \frac{a_{s-1}}{r^{s-1-i}}) \left( \frac{r+1}{2}-1 \right) - \frac{a_{s-1}}{2},  \notag \\
				m \frac{3(r+1)}{2} & \geq & 3n+ a_{s-1} \left( \left(\sum_{i=1}^{s-1} \frac{1}{r^i} \right) \left(\frac{r+1}{2}-1\right) -\frac{1}{r^{s-1}}-\frac{1}{2} \right),  \notag \\
				m \frac{3(r+1)}{2} & \geq & 3n- a_{s-1} \left(\frac{3}{2 r^{s-1}}\right). \label{Ineq12O}
				\eea
				Rewriting \eqref{Ineq2O}:
				\bea
				\rho_{s-1}+\sum_{i=0}^{s-2} a_i + p =m.  \label{Ineq13O}
				\eea
				Substituting \eqref{OddSpe},\eqref{Ineq1O} in \eqref{Ineq13O}:
				\bea
				\rho_{s-1}+\sum_{i=0}^{s-2} a_i + p &=& m, \notag \\
				\frac{a_{s-1}}{r+1}+\sum_{i=0}^{s-2} \frac{a_{s-1}}{r^{s-1-i}} & \leq & m-p, \notag \\
				a_{s-1} \leq \frac{m-p}{\frac{1}{r+1}+\sum_{i=1}^{s-1} \frac{1}{r^i}}, \notag \\
				a_{s-1} \leq \frac{(m-p)(r+1)}{1+\frac{(r^{s-1}-1)(r+1)}{(r^{s-1})(r-1)} }. \label{Ineq14O}
				\eea
				Substituting \eqref{Ineq14O} in \eqref{Ineq12O}:
				\bea
				m \frac{3(r+1)}{2} & \geq & 3n- \frac{(m-p)(r+1)}{1+\frac{(r^{s-1}-1)(r+1)}{(r^{s-1})(r-1)} } \left(\frac{3}{2 r^{s-1}}\right), \notag \\
				\text{Since $p \geq 0$, } m \frac{3(r+1)}{2} \left( 1+\frac{1}{r^{s-1}+\frac{(r^{s-1}-1)(r+1)}{(r-1)}} \right) & \geq & 3n. \label{Ineq15O} 
				\eea
				\eqref{Ineq15O} after some algebraic manipulations gives the required bound on $1-\frac{m}{n}$ and hence on $\frac{k}{n}$ as stated in the theorem. 
			\end{proof}
			Again an alternative proof for Theorem \ref{thm_todd} is given below by using linear programming:
			\begin{proof}
				The inequalities \eqref{Ineq1O},\eqref{Ineq3O},\eqref{OddSpe} and \eqref{Ineq5O} are linear inequalities and are written in matrix form as:\\
				\begin{equation*}
				A\mathbf{\underline{x}} \geq \mathbf{\underline{b}}
				\end{equation*}
				where
				\bea
				A & = & \left[
				\begin{array}{c c c c c c c c c}
					r & -1 & 0 & \hdots & 0 & 0 & 0 & 0 & 0\\
					0 & r & -1 & \hdots & 0 & 0 & 0 & 0 & 0\\
					\vdots & \vdots & \vdots & \ddots & \vdots & \vdots & \vdots & \vdots & \vdots\\
					0 & 0 & 0 & \hdots & r & -1 & 0 & 0 & 0\\
					0 & 0 & 0 & \hdots & 0 & -1 & (r+1) & -2 & (r+1)\\
					0 & 0 & 0 & \hdots & 0 & -1 & (r+1) & 0 & 0\\
					(r+3) & (r+2) & (r+2) & \hdots & (r+2) & 1 & (r+1) & 0 & (r+1)
				\end{array} \right] \label{Amat}
				\eea
				which is a $(s + 2) \times (s + 3) $ matrix and
				\bea 
				\mathbf{\underline{x}} = \left[ \begin{array}{c c c c c c c}
					a_0 & a_1 & \hdots & a_{s-1} & \rho_{s-1} & a_s & p
				\end{array} \right]^T, \mathbf{\underline{b}} = \left[ \begin{array}{c c c c c}
				0 &	0 &	\hdots & 0 & 3n
			\end{array}\right]^T
			\eea
			where $\mathbf{\underline{x}}$ is an $(s+3) \times 1$ matrix and $\mathbf{\underline{b}}$ is an $(s+2) \times 1$ matrix
			The problem of finding an upper bound on rate of the code now becomes one of minimizing $m = \mathbf{\underline{c}}^T\mathbf{\underline{x}}$, which is a linear objective function where $\mathbf{\underline{c}} = \left[ \begin{array}{c c c c c c c c}
			1 & 1 & \hdots & 1 & 0 & 1 & 0 & 1
			\end{array}\right]^T$ is an $(s+3) \times 1$ matrix, . Also by definition of $\mathbf{\underline{x}}$, $\mathbf{\underline{x}} \geq 0$. This is now in a standard form of a linear program formulation as:
			\begin{align*}
			\text{minimize } & \mathbf{\underline{c}}^T\mathbf{\underline{x}}\\
			\text{s.t. } A\mathbf{\underline{x}} & \geq \mathbf{\underline{b}}\\
			\mathbf{\underline{x}} & \geq 0
			\end{align*}
			The dual problem of the above is
			\begin{align*}
			\text{maximize } & \mathbf{\underline{b}}^T\mathbf{\underline{\lambda}}\\
			\text{s.t. } A^T\mathbf{\underline{\lambda}} & \leq \mathbf{\underline{c}}\\
			\mathbf{\underline{\lambda}} & \geq 0
			\end{align*}
			We will solve an the dual problem by writing it in standard $\text{minimize}$ $-\mathbf{\underline{b}}^T\underline{\lambda}$ form and using the simplex method. Let us introduce slack variables $h_1$,...,$h_{s+3}$ and re-write the constraints as
			\begin{align*}
			B\mathbf{\underline{v}} = \mathbf{\underline{c}}, \text{ }\mathbf{\underline{v}} \geq 0,
			\end{align*} 
			where
			\bea
			B & = & \left[
			\begin{array}{c c c c c c c c c c c c c c c c c}
				r & 0 & 0 & \hdots & 0 & 0 & 0 & (r+3) & 0 & 1 & 0 & 0 & \hdots & 0 & 0 & 0 & 0 \\
				-1 & r & 0 & \hdots & 0 & 0 & 0 & (r+2) & 0 & 0 & 1 & 0 &\hdots & 0 & 0 & 0 & 0\\
				0 & -1 & r & \hdots & 0 & 0 & 0 & (r+2) & 0 & 0 & 0 & 1 &\hdots & 0 & 0 & 0 & 0\\
				\vdots & \vdots & \ddots & \ddots & \vdots & \vdots & \vdots & \vdots & \vdots & \vdots & \vdots & \vdots & \ddots & \vdots & \vdots & \vdots & \vdots\\
				0 & 0 & 0 & \hdots & r & 0 & 0 & (r+2) & 0 & 0 & 0 & 0 & \hdots & 1 & 0  & 0 & 0\\
				0 & 0 & 0 & \hdots & -1 & -1 & -1 & 1 & 0 & 0 & 0 & 0 & \hdots & 0 & 1  & 0 & 0\\
				0 & 0 & 0 & \hdots & 0 & (r+1) & (r+1) & (r+1) & 0 & 0 & 0 & 0 & \hdots & 0 & 0 & 1 & 0\\
				0 & 0 & 0 & \hdots & 0 & -2 & 0 & 0 & 0 & 0 & 0 & 0 & \hdots & 0 & 0 & 0 & 1\\ 
				0 & 0 & 0 & \hdots & 0 & (r+1) & 0 & (r+1) & 1 & 0 & 0 & 0 & \hdots & 0 & 0 & 0 & 0\\
			\end{array} \right] \label{Amatodd},
			\eea
			and 
			\bea 
			\mathbf{\underline{v}} = \left[ \begin{array}{c c c c c c c}
				\lambda_1 & \lambda_2 & \hdots & \lambda_{s+2} & h_1 & \hdots & h_{s+3}
			\end{array} \right]^T.
			\eea
			
			With this, the objective function now is $\mathbf{\underline{d}}^T\mathbf{\underline{v}}$, where $\mathbf{\underline{d}} = \left[\begin{array}{c c c c c}
			-\mathbf{\underline{b}}^T & 0 & 0 & \hdots & 0
			\end{array} \right]^T$ which is an $(2s+5)\times 1$ matrix.Define $\mathbf{\underline{v}}_{BV}=[\lambda_1,...,\lambda_{s+2},h_1]^T$.
			
			We pick the variables $\beta_1=\lambda_1$,...,$\beta_{s+2}=\lambda_{s+2}$, $\beta_{s+3}=h_1$ as ``basic variables" and the rest, called ``non-basic variables" $\alpha_1=h_2,...,\alpha_{s+2}=h_{s+3}$ will be set to $0$. The set of basic variables is chosen such that the columns of $B$ corresponding to these basic variables is a full-rank square matrix $B_{BV}$. The remaining columns of $B$ will give a matrix $B_{NBV}$. The system of equations is now in the following form:
			\begin{align*}
			\left[\begin{array}{c c}
			B_{BV} & B_{NBV}
			\end{array}\right]
			\left[\begin{array}{c}
			\mathbf{\underline{v}}_{BV} \\
			\hline
			\mathbf{\underline{0}}
			\end{array}
			\right] = \mathbf{\underline{c}}
			\end{align*}
			Therefore we will equivalently solve
			\begin{align*}
			B_{BV}\mathbf{\underline{v}}_{BV} = \mathbf{\underline{c}}
			\end{align*}
			The above system of equations can be solved in closed form to get the following:
			\begin{align}
			\lambda_{s+2} & = \frac{2\sum_{i = 1}^{s-1}r^i+1}{3(r^{s}+2\sum_{i = 1}^{s-1}r^i+1)},\label{Odd1}\\
			h_1 & = \frac{\sum_{i = 1}^{s-1}r^i+2}{3(\sum_{i = 0}^{s-1}r^i)},\label{Odd2}\\
			\lambda_{j+1} & = \frac{r^s-3r^{s-n-1}+r+1}{3(r+1)(r^s-1)},\text{   for }0 \leq j \leq s-2 \label{Odd3}\\
			\lambda_s & = 0, \label{Odd4}\\
			\lambda_{s+1} & = \frac{\sum_{i = 1}^{s-1}r^i+2}{3(r+1)(\sum_{i = 0}^{s-1}r^i)}.\label{Odd5}
			\end{align}
			which are non-negative if $r \geq 3$.
			
			 Hence the solution given by \eqref{Odd1},\eqref{Odd2},\eqref{Odd3},\eqref{Odd4} and \eqref{Odd5} is a basic feasible solution. Let the elements of the vector $\mathbf{\underline{d}}$ be indexed by the elements of the vector $\mathbf{\underline{v}}$ i.e., $i^{th}$ component of $\mathbf{\underline{d}}$ is indexed by $i^{th}$ component of $\mathbf{\underline{v}}$. To check for optimality we check if the ``reduced cost coefficients" $r_{\alpha_i} = d_{\alpha_i} - z_{\alpha_i}$ are non-negative, for every non-basic variable $\alpha_i,1\leq i \leq s+2$. We note that for the above made choice of basic and non-basic variables, in the vector $\mathbf{\underline{d}}$ only $d_{\beta_{s+2}} = -3n$ is non-zero. The quantity $z_{\alpha_i}$ is defined as follows:
			 
			\begin{equation*}
			z_{\alpha_i} = \sum_{j = 1}^{s+3}d_{\beta_j}y_{(j,\alpha_i)} = d_{\beta_{s+2}}y_{(s+2,\alpha_i)} = -3ny_{(s+2,\alpha_i)}
			\end{equation*}
			where $y_{(s+2,\alpha_i)}$ are as shown in the row reduced echelon form of matrix B below:
			\begin{align*}
			B_{rref} = \left[\begin{array}{c c c c c c c c c c c c c}
			1 & 0 & 0 & \hdots & 0 & 0 & 0 & y_{(1,\alpha_1)} & y_{(1,\alpha_2)} & y_{(1,\alpha_3)} & \hdots & y_{(1,\alpha_{s+1})} & y_{(1,\alpha_{s+2})} \\
			0 & 1 & 0 & \hdots & 0 & 0 & 0 & y_{(2,\alpha_1)} & y_{(2,\alpha_2)} & y_{(2,\alpha_3)} & \hdots & y_{(2,\alpha_{s+1})} & y_{(2,\alpha_{s+2})}\\
			0 & 0 & 1 & \hdots & 0 & 0 & 0 & y_{(3,\alpha_1)} & y_{(3,\alpha_2)} & y_{(3,\alpha_3)} & \hdots & y_{(3,\alpha_{s+1})} & y_{(3,\alpha_{s+2})}\\
			\vdots & \vdots & \vdots & \ddots & \vdots & \vdots & \vdots & \vdots & \vdots & \vdots & \ddots & \vdots & \vdots\\
			0 & 0 & 0 & \hdots & 1 & 0 & 0 & y_{(s+1,\alpha_1)} & y_{(s+1,\alpha_2)} & y_{(s+1,\alpha_3)} & \hdots & y_{(s+1,\alpha_{s+1})} & y_{(s+1,\alpha_{s+2})}\\
			0 & 0 & 0 & \hdots & 0 & 1 & 0 & y_{(s+2,\alpha_1)} & y_{(s+2,\alpha_2)} & y_{(s+2,\alpha_3)} & \hdots & y_{(s+2,\alpha_{s+1})} & y_{(s+2,\alpha_{s+2})}\\
			0 & 0 & 0 & \hdots & 0 & 0 & 1 & y_{(s+3,\alpha_1)} & y_{(s+3,\alpha_2)} & y_{(s+3,\alpha_3)} & \hdots & y_{(s+3,\alpha_{s+1})} & y_{(s+3,\alpha_{s+2})}\\ 
			\end{array}
			\right]
			\end{align*}
			
			We observe that in reducing $B$ to $B_{rref}$, to row-$(s+2)$ only non-negative linear combinations of the rows above it are added, entries of which are either $0$ or $1$. Therefore $r_{\alpha_i} \geq 0$ for $\alpha_i$ all non-basic variables. Hence $y_{(s+2,\alpha_1)},...,y_{(s+2,\alpha_M)} \geq 0$. Hence the basic solution is an ``optimal basic feasible" solution.\\
			By the theorem of strong duality the optimal solutions of the primal problem and the dual problem are equal. Therefore the minimum value of $m$ is $3n\lambda_{s+2} = \frac{n(2\sum_{i = 1}^{s-1}r^i+1)}{(r^{s}+2\sum_{i = 1}^{s-1}r^i+1)}$.\\
			Hence we get the upper bound on the rate:
			\begin{equation*}
			\frac{k}{n} \leq 1 - \frac{m}{n} \leq \frac{r^s}{(r^{s}+2\sum_{i = 1}^{s-1}r^i+1)}
			\end{equation*}
			We now pick a solution for the primal problem and show that it is feasible and gives the optimal objective function value.
			\begin{align*}
			a_i & = \frac{nr^i(r+1)}{r^{s}+2\sum_{i = 1}^{s-1}r^i+1},\text{ for }0 \leq i \leq s-1\\
			\rho_{s-1} & = \frac{nr^{s-1}}{r^{s}+2\sum_{i = 1}^{s-1}r^i+1},\text{   }a_{s} = 0,\text{   }
			p = 0.
			\end{align*}
			It is easy to check that this solution satisfies the constraints of the primal problem with equality. 
%			It remains to check the following:
%			\begin{equation*}
%			(r+3)a_0+(r+2)\sum_{i = 1}^{\frac{t}{2}-1}a_i+a_{\frac{t}{2}} \geq 3n.
%			\end{equation*}
%			Upon simplification, it is seen that the above is met with equality. 
Therefore the chosen solution is a feasible solution. It is also easy to check that the solution gives the optimal value of the objective function. Hence it is an optimal feasible solution. We thus conclude that a code having the above chosen values will have the optimal rate.
		\end{proof}
			
			From the two derivations, it becomes clear that a code achieving the rate bound given in \eqref{Thm1} will have parity check matrix of the form \eqref{Hform} with $D$ being an empty matrix and all the inequalities  \eqref{Ineq1},\eqref{Ineq2},\eqref{Ineq3},\eqref{Ineq5} met with equality. A similar observation holds for $t$ odd also.
			It may be noted here that our bound, for the special cases of $t=2,3,4$, matches with the rate bound given in \cite{PraLalKum},\cite{song2016sequential},\cite{balaji2016binary} respectively. In the rest of the paper, the codes achieving the bounds \eqref{Thm1} or \eqref{Thm2} depending on $t$ will be referred to as ``rate-optimal codes".
%			
%			
%			Note that our bound \eqref{Thm2}, for the special case of $t=3$, matches with the rate bound given in \cite{song2016sequential}.
\begin{note}
	We now make a remark on the blocklength of the rate-optimal codes. From the proofs it can be seen that, for $t$ even, for the optimal values of $a_0,...,a_{\frac{t}{2}}$ to be integral, $2n$ needs to be an integer multiple of $r^{\frac{t}{2}} + 2 \sum_{i=0}^{\frac{t}{2}-1} r^i$. Similarly for $t$ odd, for the optimal values of $a_0,...,a_{s-1},\rho_{s-1}$ to be integral, $(r+1)n$ needs to be an integer multiple of $r^{s}+2\sum_{i = 1}^{s-1}r^i+1$ and $nr^{s-1}$ needs to be an integer multiple of $r^{s}+2\sum_{i = 1}^{s-1}r^i+1$.
\end{note}
\begin{note}
 The linear program was numerically solved for integer values of the variables and objective function. The solution for the case of $t = 4$ and $r = 5$ is shown in figure \ref{fig:IntLinProg}. Notice that integer solutions achieving our rate bound are possible for a set of blocklengths in a periodic fashion as predicted by theoretical solution.
	\begin{figure}[ht]
		\centering
		\includegraphics[scale = 0.3]{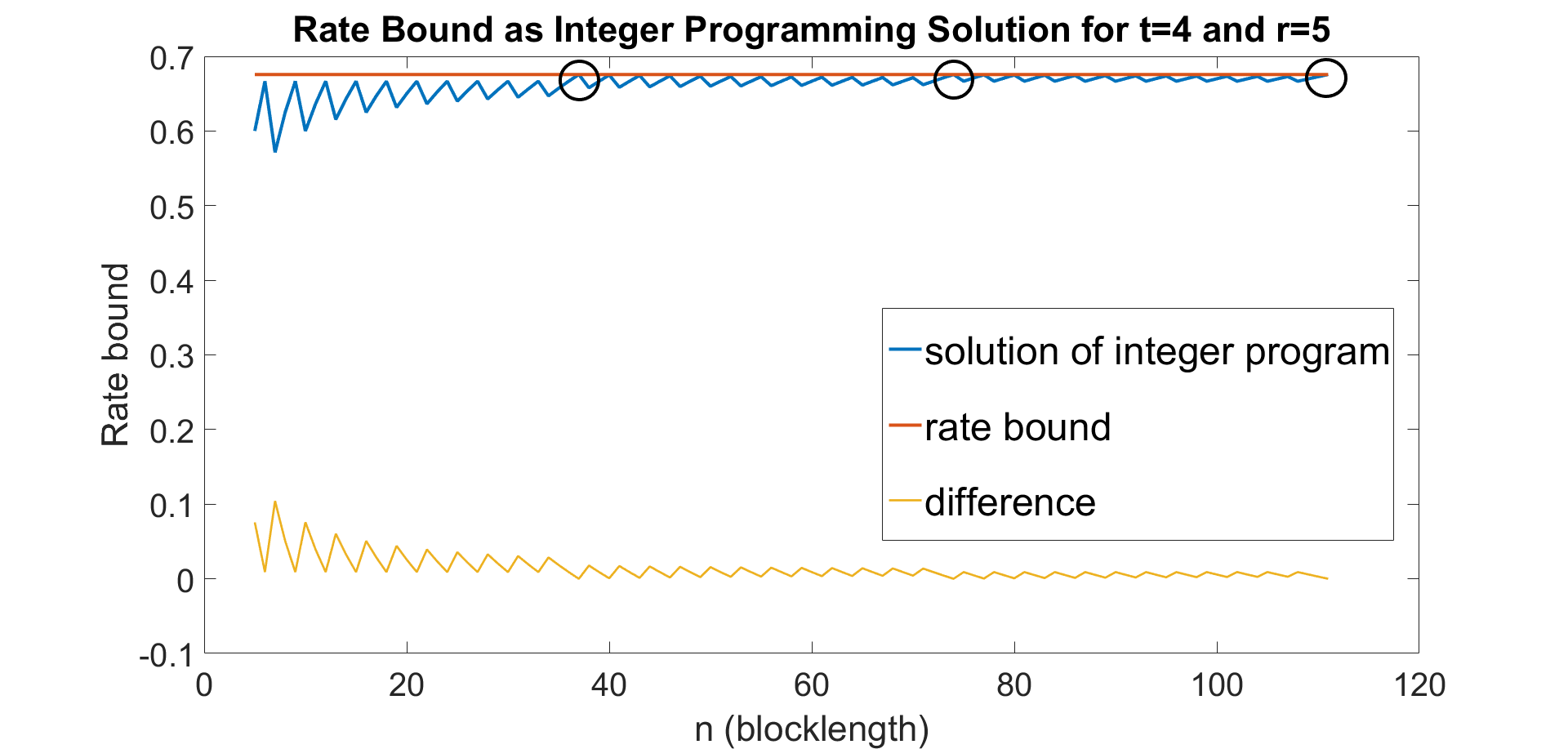}
		\caption{Plot showing integer solution for the linear program formulation. Circled points are the values of blocklength that allow for integer solution achieving our rate bound.}
		\label{fig:IntLinProg}
	\end{figure}
\end{note}
\begin{figure}[ht]
	\centering
	\includegraphics[width=4in]{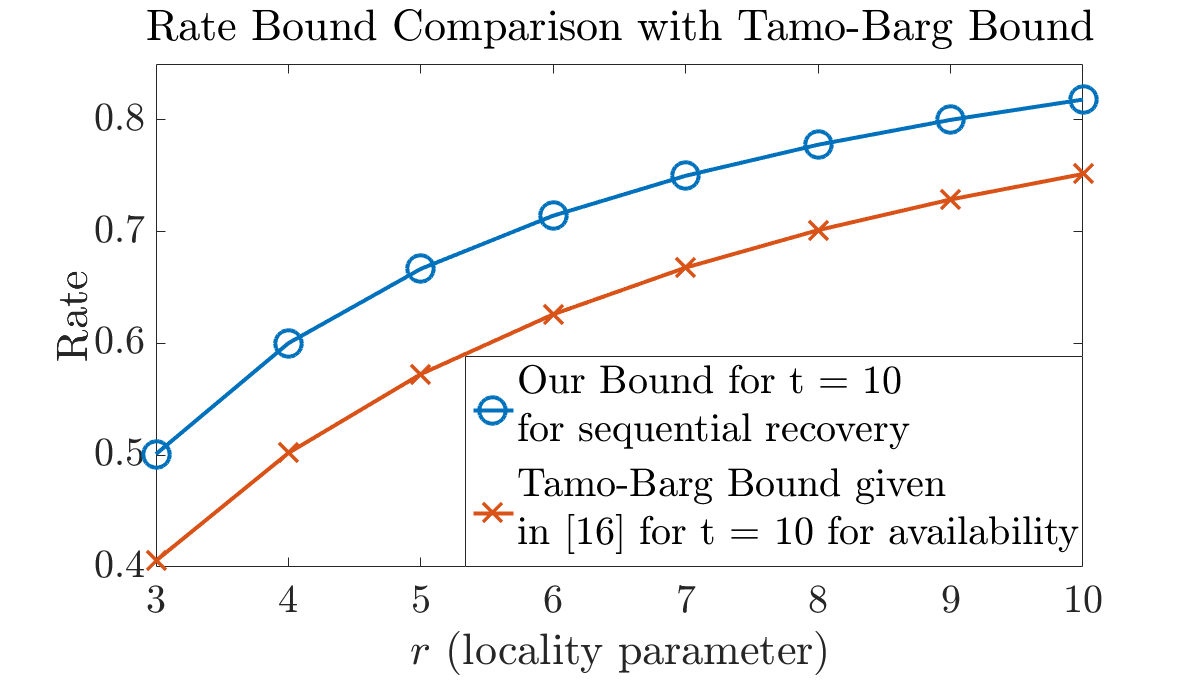}
	\caption{Comparison of rate bounds on codes with availability and codes with sequential recovery for an example case $t=10$.}
	\label{fig:Comparison_Tamo}
\end{figure}
\begin{note}
	In general, one would expect an $(n,k,r,t)_{seq}$ code to achieve a higher rate than a counterpart code having $t$-availability. The achievable upper bound on rate of a code with sequential recovery is significantly larger than the upper bound on the rate of codes having availability given in \cite{TamBarFro} as shown in Fig.~\ref{fig:Comparison_Tamo}.
\end{note}

\section{Construction of codes achieving the upper bound on rate}
In this section we give a construction of codes with sequential recovery achieving the rate bound \eqref{Thm1} for any $r \geq 3,t=2s,s \geq 1$ and also give a construction achieving \eqref{Thm2} for any $r \geq 3,t=2s-1,s \geq 1$.
\subsection{Construction of Binary Rate-Optimal Codes for $t$ Even}
%In this subsection we give a construction of codes achieving the rate bound \eqref{Thm1} for any $r \geq 3$ and $t \in 2\mathbb{Z_{+}}$. A construction achieving the bound \eqref{Thm1} for the special case of $t=2$ was provided in \cite{PraLalKum}, \cite{SonYue_3_Erasure}, \cite{BalPraKum} and  for the case of $t = 4$ was provided in \cite{balaji2016binary}.\\
In this subsection we give a construction of rate-optimal codes for any $r \geq 3$ and $t$-even. While the rate bound is independent of the size of the field over which the code is defined, our construction provides optimal binary codes. A construction achieving the bound \eqref{Thm1} for the special case of $t=2$ was provided in \cite{PraLalKum}, \cite{SonYue_3_Erasure}, \cite{BalPraKum} and  for the case of $t = 4$, in \cite{balaji2016binary}.

We provide an iterative, graph-based construction of rate-optimal codes for a given $r \geq 3$ and $t$ even. We build a graph $G_{\frac{t}{2}-1}$ iteratively, starting from $G_0$, by adding nodes in vertical, layer-by-layer fashion.  At every stage of the iteration, the graph $G_i$ thus constructed, will always have girth $\geq t+1$. We then define our rate-optimal code based on the graph $G_{\frac{t}{2}-1}$. \\

\textbf{Construction of graph $G_{\frac{t}{2}-1}$:}
Denote the vertex set of a graph $G$ by $V(G)$. Let $N(v)$ denote the neighbors of a node $v$ in a graph. We make use of graphs $G_0$ and $B_i$, for $1 \leq i \leq \frac{t}{2}-1$ described below as the ingredients:
\ben
%\item We construct $G_0$ as follows: Let $G$ be an $r$-regular graph with girth $\geq t+1$\textsuperscript{\ref{Fur}}. The construction requires the base-graph $G_0$ to be $r$-regular with girth $\geq t+1$ and number of nodes that is a multiple of $r^{\frac{t}{2}-1}$. Hence we take the disjoint union of a number of 
%%$\frac{\text{LCM}(|V(G)|,r^{\frac{t}{2}-1})}{|V(G)|}$ 
%copies of $G$, with the number chosen to ensure that the size of the resultant graph is a multiple of $r^{\frac{t}{2}-1}$.  We then declare the resultant graph to be $G_0$. Define $U_0=V(G_0),|U_0|=u_0$.  
\item We pick $G_0$ as follows: Let $G_0$ be an $r$-regular graph with girth $\geq t+1$\textsuperscript{\ref{Fur}}. Define $U_0=V(G_0),|U_0|=u_0$.  
\item Next, for $i$ in the range $1 \leq i \leq \frac{t}{2}-1$, we iteratively pick an $(r,u_{i-1})$-biregular bipartite graph $B_i$ with girth $g_i \geq \left \lceil \frac{t+1}{i+\frac{1}{2}} \right \rceil$ \footnote[2]{Such a graph can be constructed, due to \cite{Furedi}\label{Fur}}. Let $U_i,L_i$ be the two (upper and lower) sets of nodes in the bipartite graph with $\text{degree}(x)=r,\forall\ x\in U_i$ and $\text{degree}(y)=u_{i-1},\forall\ y\in L_i$. Hence $V(B_i)=U_i \cup L_i$. and Let $|U_i|=u_i$ and $|L_i|=l_i$. By edge counting, we have $ru_i = u_{i-1}l_i$.  
\een
\begin{figure}
	\centering 
	\tikzstyle{decision} = [diamond, draw, fill=gray!20, 
	text badly centered, node distance=2cm, inner sep=0pt, aspect=2.5]
	
	\tikzstyle{block} = [rectangle, draw, fill=gray!20, 
	text width=22em, text centered, rounded corners, node distance=2.4cm]
	
	\tikzstyle{block2} = [rectangle, draw, fill=gray!20, 
	text centered, rounded corners, node distance=1.4cm]
	
	\tikzstyle{block3} = [rectangle, draw, fill=gray!20,  text centered, rounded corners, node distance=1.4cm]    
	
	\tikzstyle{line} = [draw, -latex']
	
	\tikzstyle{cloud} = [draw, ellipse,fill=red!20, node distance=2cm]
	
	\begin{tikzpicture}[node distance = 1.8cm, auto]
	% Place nodes
	\node [block] (INIT) {\small{Pick $G_0$, $r$-regular, girth $\geq(t+1)$,\\ $V(G_0)=U_0$,\ $|U_0|=u_0$}};
	
	\node [block3, below of=INIT, yshift=0.2cm] (INITIALI) {\small{$i=1$}} ;
	
	\node [block, below of=INITIALI, yshift=1cm] (BIPARTITE) {\small{Pick bipartite graph $B_i$: $(r,u_{i-1})$-biregular, \\ $V(B_i)=U_i \cup L_i$, girth $\geq  \left \lceil (t+1)/(i+1/2) \right \rceil$, \\ $|U_i|=u_i,\  |L_i|=l_i$}};
	
	\node [block, below of=BIPARTITE, yshift=-0.4cm] (SPLIT) {\small{Let $L_i=\{v_1^i,v_2^i,...,v_{l_i}^i\}$ \\ For $1 \leq j \leq l_i$, \\ Split the node $v_j^i \in L_i$ in $B_i$ with $\text{deg}(v_j^i)=u_{i-1}$, into $u_{i-1}$ degree-$1$ nodes $(V_j^i)$} such that $N(v_j^i)=N(V_j^i)$\\ ($V_j^i$ is the set of $u_{i-1}$ degree 1 nodes which are formed by splitting $v_j^i$)\\An example of splitting a node into degree-1 nodes is shown in fig \ref{fig:Splitting}};
	
	\node [block, below of=SPLIT, yshift=-0.7cm] (REPLICATE) {\small{Replicate the graph $G_{i-1}$ $l_i$ times (thus each upper node in $U_{i-1}$ is also replicated $l_i$ times), \\ $(G_{i-1})_j$: $j^{\text{th}}$ copy of $G_{i-1}$ and $(U_{i-1})_j$: copy of $U_{i-1}$ in $(G_{i-1})_j$ \\ $(U_{i-1} \subseteq V(G_{i-1}))$}};
	
	\node [block2,  left of=REPLICATE, node distance=4cm, xshift=-1cm] (UPDATE) {\rotatebox[]{90}{$i=i+1$}};
	
	\node [block, below of=REPLICATE, yshift=-0.3cm] (MERGE) {\small{Let $(U_{i-1})_j=\{(c_{1}^{i-1})_j,...,(c_{u_{i-1}}^{i-1})_j\}$ and $V_j^i=\{(d_{1}^i)_j,...,(d_{u_{i-1}}^i)_j \}$ \\ Merge the node $(c_{m}^{i-1})_j \in (U_{i-1})_j$ with $(d_{m}^i)_j \in V_j^i$, \\ $\forall\ 1 \leq m \leq u_{i-1}$ and $\forall\ 1 \leq j \leq l_i$ \\ An example of merging two nodes is shown in fig \ref{fig:Merging}}};
	
	\node[block, below of=MERGE, yshift=-0.1cm] (GIGRAPH) {\small{The resulting graph is $G_i$; can be verified that $G_i$ has girth $\geq t+1$, the nodes $U_i \subset V(G_i)$ now form the upper layer of the graph $G_i$ and these are the nodes in $G_i$ that participate in the next iterative step}};
	
	\node [decision, below of=GIGRAPH] (DECIDE) {is $i = \frac{t}{2}-1?$};
	
	\node [block3, below of=DECIDE] (stop) {stop};
	
	% Draw edges
	\path [line] (INIT) -- (INITIALI);
	\path [line] (INITIALI) -- (BIPARTITE);
	\path [line] (BIPARTITE) -- (SPLIT);
	\path [line] (SPLIT) -- (REPLICATE);
	\path [line] (REPLICATE) -- (MERGE);
	\path [line] (MERGE) -- (GIGRAPH);
	\path [line] (GIGRAPH) -- (DECIDE);
	\path [line] (DECIDE) -| node [near start] {NO} (UPDATE);
	\path [line] (UPDATE) |- (BIPARTITE);
	\path [line] (DECIDE) -- node {YES}(stop);
	%   \path [line,dashed] (expert) -- (init);
	%   \path [line,dashed] (system) -- (init);
	%   \path [line,dashed] (system) |- (replicatebi);
	
	\end{tikzpicture}
	
	\caption{Flowchart showing the iterative construction of $G_{\frac{t}{2}-1}$ for $t\geq 4$. For $t=2$, $G_{\frac{t}{2}-1}=G_0$. Hence the graph $G_0$ mentioned in the first step of the flow chart will be our required graph construction for $t=2$.} 
	\label{fig:Flowchart}
\end{figure}
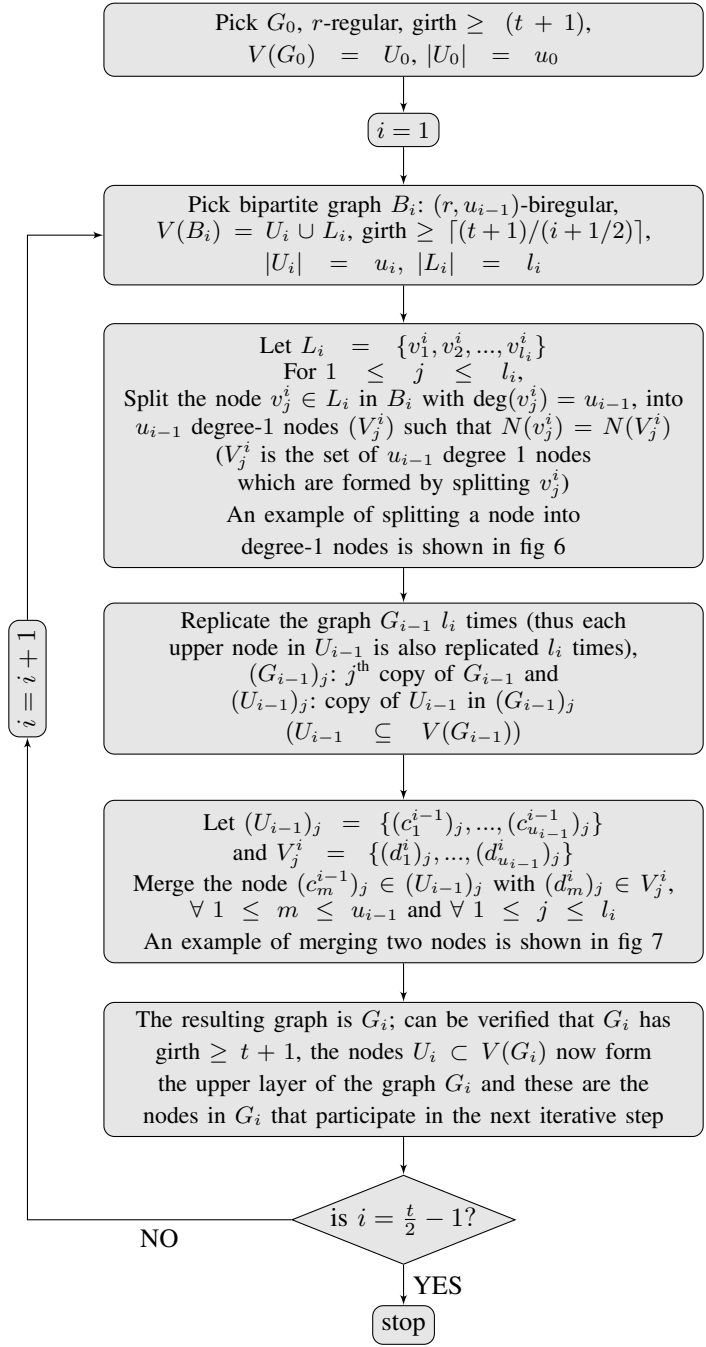
\begin{figure}[ht]
	\centering
	\includegraphics[scale = 0.8]{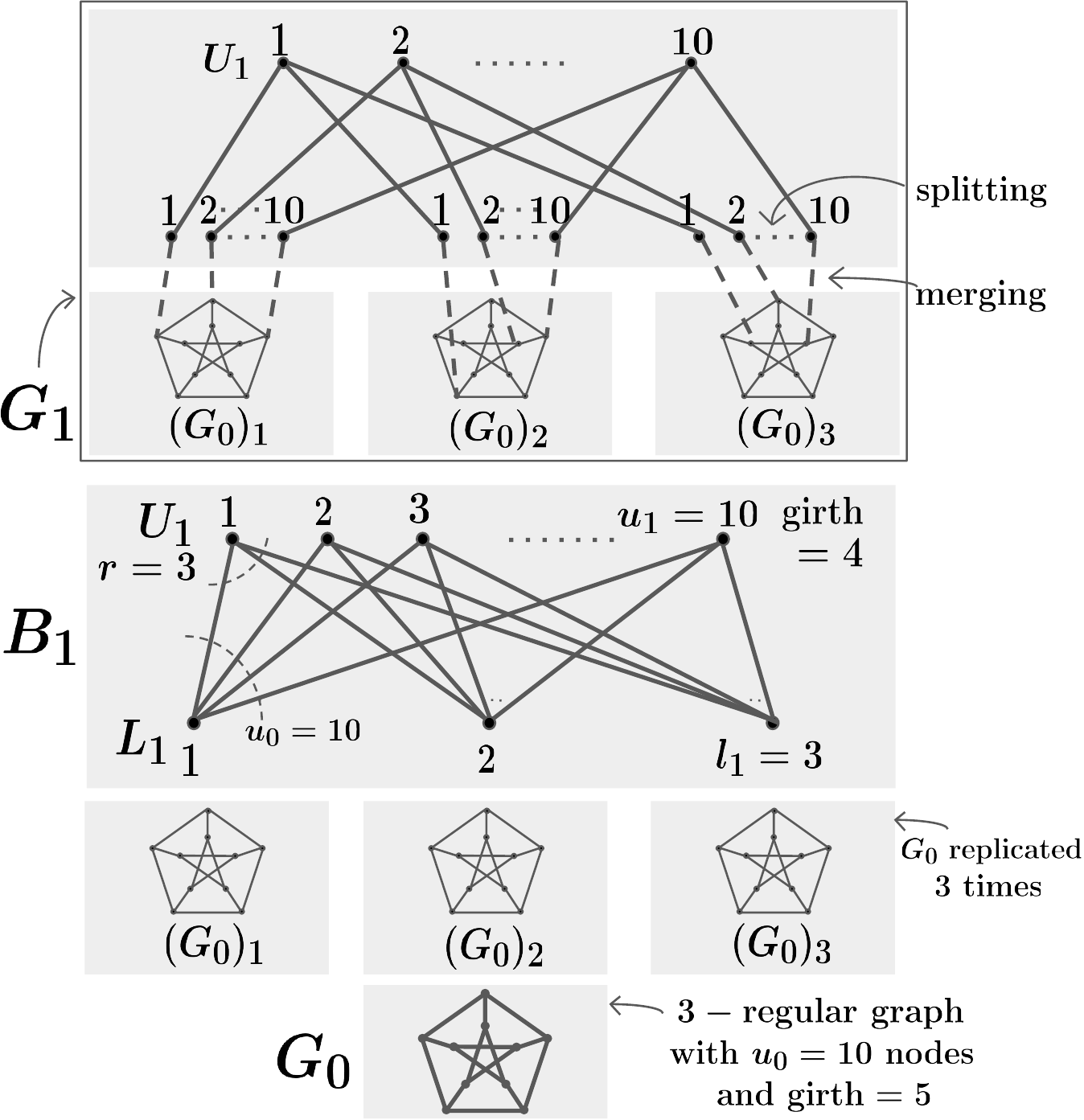}
	\caption{Depiction of steps mentioned in the flowchart Fig \ref{fig:Flowchart} for an example construction of $G_1$ for $t = 4$, $r = 3$. Here $G_0$ is the Petersen graph. For the purpose of representation the merger of nodes is shown only for 9 nodes by the dashed lines.}
	\label{fig:ExampleConstruction}
\end{figure}
\begin{figure}[ht]
	\centering
	\includegraphics[scale = 0.8]{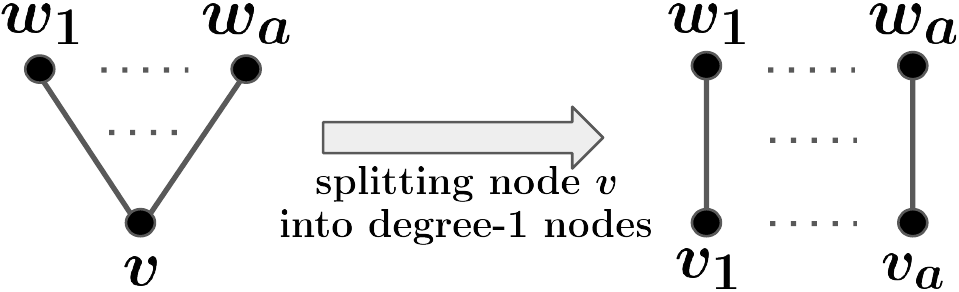}
	\caption{Example of splitting a node into degree-1 nodes}
	\label{fig:Splitting}
\end{figure}
\begin{figure}[ht]
	\centering
	\includegraphics[scale = 0.6]{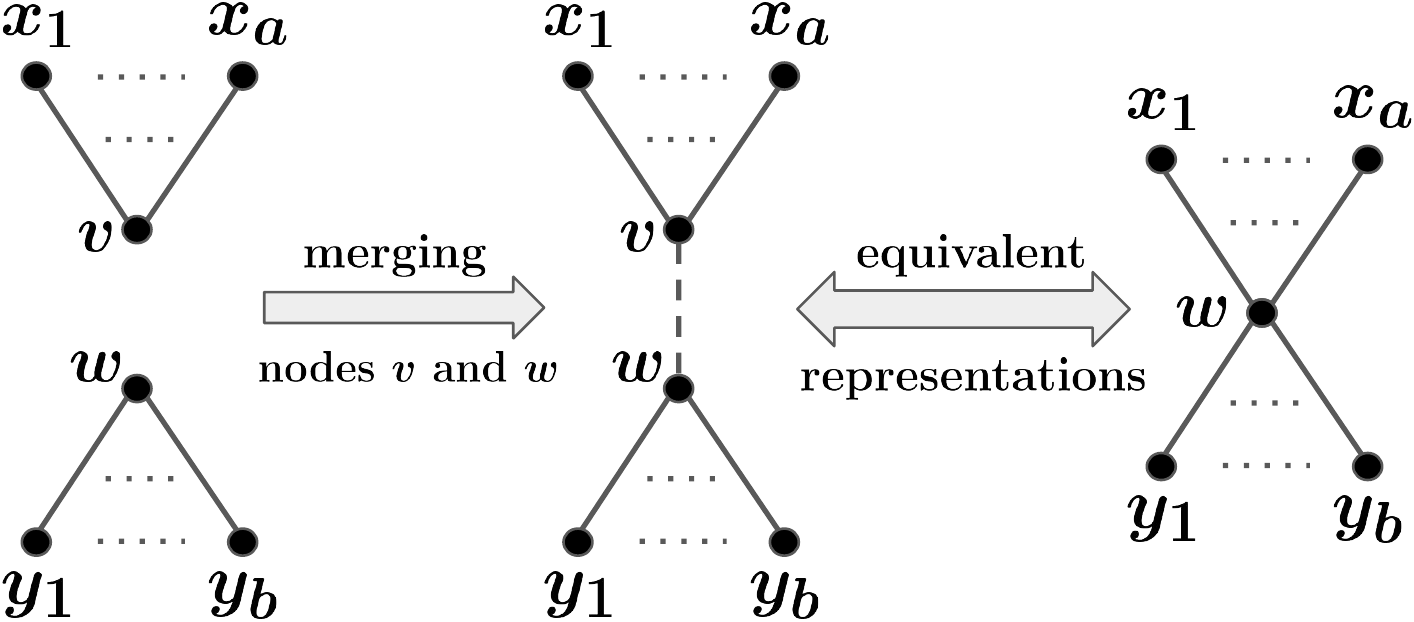}
	\caption{Example of merging two nodes}
	\label{fig:Merging}
\end{figure}
The flowchart shown in Fig.~\ref{fig:Flowchart} describes the iterative construction of the graph $G_{\frac{t}{2}-1}$. An example construction of $G_1$ for $t=4,r=3$ depicting the steps of the flowchart for the case of a single iteration appears in Fig.~\ref{fig:ExampleConstruction}.
We make the following observations:
\begin{enumerate}
	\item For $0 \leq i \leq \frac{t}{2}-1$, there are $f_i=\prod_{j=i+1}^{\frac{t}{2}-1}l_j$ disjoint copies of $G_i$ that reside inside $G_{\frac{t}{2}-1}$ as subgraphs due to the replication steps in the construction of $G_{\frac{t}{2}-1}$ (given in the flowchart Fig \ref{fig:Flowchart}).  For all $\frac{t}{2}-1 > j>i$, since copies of $G_i$ reside inside $G_j$ as subgraphs, when we replicate $G_j$ in the $(j+1)^{th}$ replication step, all the copies of $G_i$ contained within $G_j$ also get replicated; we refer to all such copies of $G_i$ here. Let us denote the disjoint union of these subgraphs corresponding to all copies of $G_i$ formed due to replication steps by $G_i^{rep}=\cup_{j=1}^{f_i} (G_i)_j$, where $(G_i)_j$ is the $j^{th}$ copy of $G_i$ in $G_{\frac{t}{2}-1}$ formed due to some replication step. We view $G_i^{rep}$ as a subgraph of $G_{\frac{t}{2}-1}$.
	\item Note that the construction of $G_{\frac{t}{2}-1}$ proceeds by adding nodes layer by layer with each layer connecting to the layer below it in a tree-like fashion while maintaining the girth of the overall graph to be atleast $t+1$. Let $N(v)$ represent neighbors of a node $v$ in $G_{\frac{t}{2}-1}$. Let $(U_i)_j$ be the copy of $U_i$ in $V((G_i)_j), \forall 1 \leq j \leq f_i$. For $1 \leq i \leq \frac{t}{2}-1 $, the nodes $T_i=\cup_{j=1}^{f_i} (U_i)_j$ represent the nodes in the $i^{th}$ layer and they connect in a tree-like fashion to the $(i-1)^{th}$ layer nodes $T_{i-1}=\cup_{j=1}^{f_{i-1}} (U_{i-1})_j$. The connection is tree-like because for $v \in T_i$, let $S_v = N(v) \cap T_{i-1}$ then $S_v \cap S_w = \emptyset$ and $|S_v|=|S_w|=r$ $,\forall v,w \in T_i,  v \neq w$ and $N(v) \cap T_p = \emptyset ,\forall p \notin \{i-1,i+1\}$,$\forall v \in T_i$. 
	\item \textbf{Girth of $G_i$:} Let the girth of $G_{i-1}$ be at least $t+1$. Consider the construction of graph $G_i$ from $l_i$ replicated copies of $G_{i-1}$ and the biregular bipartite graph $B_i$. $(U_{i-1})_j$ is the copy of $U_{i-1}$ in the $j^{\text{th}}$ copy of $G_{i-1}$ i.e., $(G_{i-1})_j$, $1\leq j \leq l_i$ appearing at the replication step in the construction of $G_i$ from $G_{i-1}$. Consider a cycle in $G_i$. If the cycle does not involve any nodes in $U_i$ then it must be completely contained in $(G_{i-1})_j$ for some $1\leq j \leq l_i$ and hence of length at least $t+1$ because the girth of $G_{i-1}$ is assumed to be at least $t+1$. If the cycle involves a node $v_1 \in U_i$ then the cycle must be of the form $v_1P_{a_1b_1}v_2...v_qP_{a_qb_q}v_1$ (as shown in fig \ref{fig:GirthExplanation}) where $v_m \in U_i, a_m,b_m \in \bigcup \limits_{j=1}^{l_i}(U_{i-1})_j$ and $P_{a_mb_m}$ (shown in green in fig \ref{fig:GirthExplanation}) is a path in $(G_{i-1})_p$ for some $1\leq p\leq l_i$ starting at $a_m$ and ending at $b_m$ for $1\leq m\leq q$. This is because when we start from $v_1$ and move to a node say $a_1$ in $\bigcup \limits_{j=1}^{l_i}(U_{i-1})_j$ the only way to come back to $v_1$ for completing the cycle is by traversing a path $P_{a_1b_1}$ in $(G_{i-1})_p$ for some $1\leq p\leq l_i$ ending at some node $b_1 \in (U_{i-1})_p$. By repeating the above argument it can be seen that the cycle must be of the form $v_1P_{a_1b_1}v_2...v_qP_{a_qb_q}v_1$. Now by above argument,for $1\leq m\leq q$, let $a_m,b_m \in (U_{i-1})_{\sigma_m}$, $1 \leq \sigma_m \leq l_i$. Hence $a_m,b_m$ must have been merged with some two nodes in $V^i_{\sigma_m}$ (nodes coming from splitting $v^i_{\sigma_m} \in L_i$ into degree 1 nodes) in the construction of $G_i$. Hence $v_1v^i_{\sigma_1}v_2v^i_{\sigma_2} \hdots v_qv^i_{\sigma_q} v_1$ must contain a non trivial cycle in $B_i$. Also the length of the path $P_{a_mb_m}$ is at least $2(i-1)+1$ due to the tree-like structure of $(G_{i-1})_{\sigma_m}$. Hence if the girth of $B_i$ is $g_i$ then the length of the cycle $v_1P_{a_1b_1}v_2...v_qP_{a_qb_q}v_1$ is at least $2q+q(2(i-1)+1) \geq g_i+\frac{g_i}{2}(2(i-1)+1)= g_i(i+\frac{1}{2}) \geq \frac{t+1}{i+\frac{1}{2}}(i+\frac{1}{2}) = t+1$. Hence girth of $G_i$ is at least $t+1$.
%	
%	Consider construction of graph $G_i$ from replicated copies of $G_{i-1}$ and the biregular bipartite graph $B_i$. Let the girth of $B_i$ be $g_i$. Consider a smallest cycle in $B_i$. After splitting the nodes in $L_i$ into degree-1 nodes and merging with nodes in copies of $U_{i-1}$ we will compute the length of a smallest possible cycle in $G_i$. In fig \ref{fig:GirthExplanation} we show only those edges of the graph $G_i$ that are part of a smallest cycle in $G_i$. Note that the edges above the nodes labeled by $\bigcup_j (U_{i-1})_j$ (colored blue) constituted a smallest cycle of length $g_i$ in $B_i$. Now, let the edges below the nodes labeled by $\bigcup_j (U_{i-1})_j$ (colored green) represent the part of the smallest cycle that lies in copies of $G_{i-1}$. Consider two nodes $a,b \in \bigcup_j (U_{i-1})_j$. Due to the tree-like structure of $G_{i-1}$ it can be seen that the length of the path in the smallest cycle between nodes $a$ and $b$ below them (shown by green edges) is at least $2(i-1)+1$. Thus, for every $2$ edges in $B_i$ an additional length of $2(i-1)+1$ gets added to the cycle length. Therefore it is seen that a smallest cycle in $G_i$ will have a length of at least $\frac{g_i}{2}(2i+1)$. Thus, if $g_i \geq \left \lceil \frac{t+1}{i+\frac{1}{2}} \right \rceil$ then girth of $G_i$ is at least $t+1$.
\end{enumerate}
\begin{figure}[ht]
	\centering
	\includegraphics[scale = 0.4]{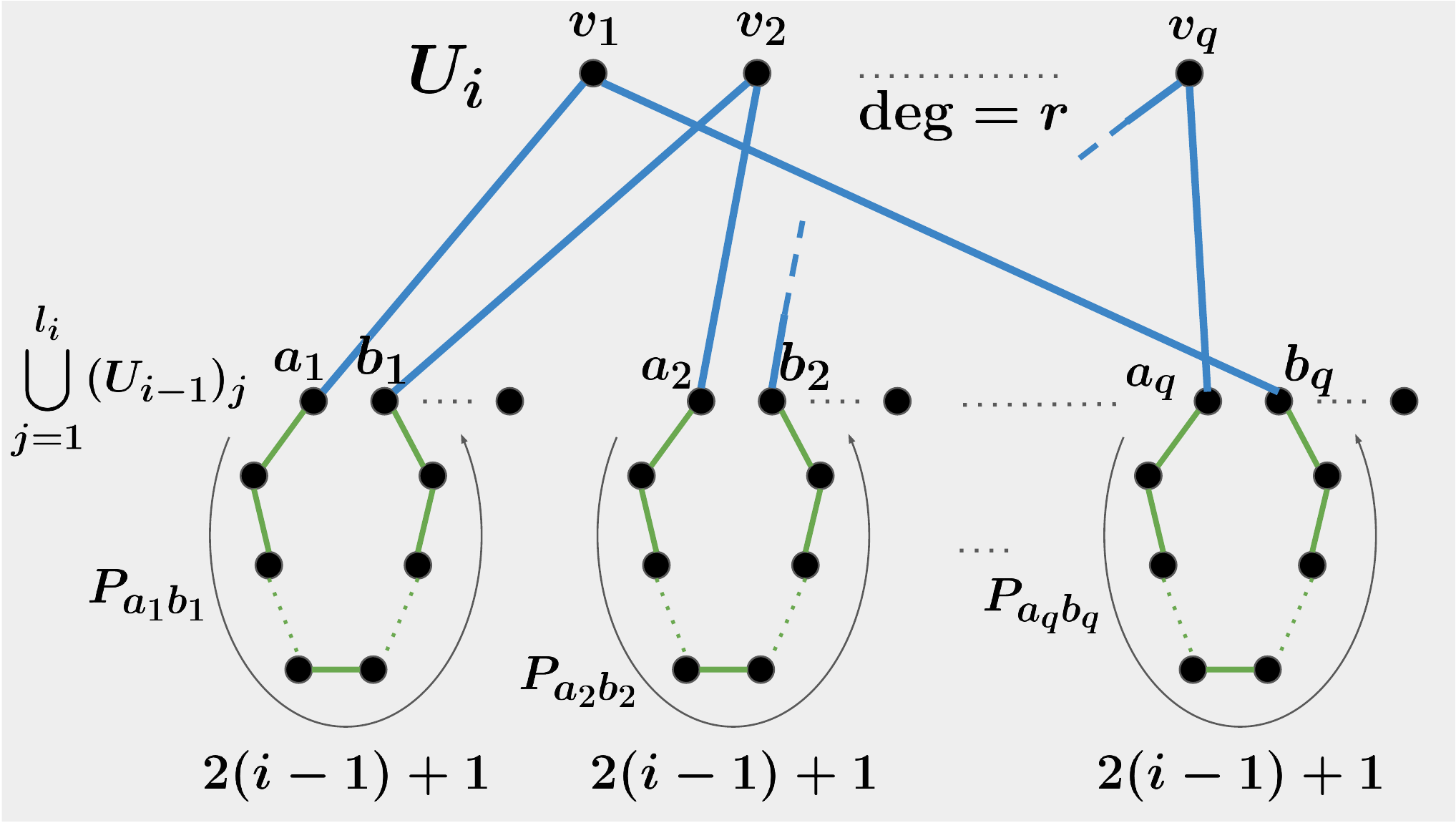}
	\caption{Figure showing a cycle in the graph $G_i$. Only the edges constituting a cycle are shown.}
	\label{fig:GirthExplanation}
\end{figure}
\vspace{0.1cm}
\textbf{Description of a rate-optimal code $\mathcal{C}$ based on $G_{\frac{t}{2}-1}$:}\\
Code $\mathcal{C}$ is defined on the graph $G_{\frac{t}{2}-1}$ as follows:
\begin{enumerate}
	\item The edges of $G_0^{rep}$ represent information symbols. It can be seen that the number of information symbols is $k = \frac{u_0r}{2}\prod_{j=1}^{\frac{t}{2}-1}l_j$.
	\item Every node of $G_{\frac{t}{2}-1}$ represents a distinct parity symbol described as follows.
	\item A node $v \in V(G_0^{rep})$ represents a parity symbol which is the binary sum of information symbols that are represented by edges in $G_0^{rep}$ incident on $v$. 
	\item For $1 \leq i \leq \frac{t}{2}-1$, a node $v \in T_i$ in $G_{\frac{t}{2}-1}$ represents a parity symbol which is the binary sum of code symbols that are represented by the nodes in $N(v) \cap T_{i-1}$.
	
	%	every node in the sets $U_{i}$ of the copies of $G_{i}$ in the graph $G_i^{rep}$ represents a parity symbol which is the binary sum of symbols represented by its neighboring nodes in $G_i^{rep}$.\\
	%	 Equivalently, a parity symbol represented by a node $v \notin V(G_0^{rep})$ in $G_{\frac{t}{2}-1}$ is the binary sum of code symbols represented by neighbors of $v$ which are introduced before $v$ in the iterative construction of $G_{\frac{t}{2}-1}$ represented by the flow chart Fig \ref{fig:Flowchart}.
	\item It can be seen that the total number of nodes is $n-k = \sum_{i=0}^{\frac{t}{2}-1}u_i\prod_{j=i+1}^{\frac{t}{2}-1}l_j = \sum_{i=0}^{\frac{t}{2}-1}\frac{u_0}{r^i}\prod_{j=1}^{\frac{t}{2}-1}l_j$.
	\item  $\mathcal{C}$ is defined by the information symbols represented by edges in  $G_0^{rep}$ and parity symbols represented by nodes in $G_{\frac{t}{2}-1}$.
	\item From the above counts, the rate of the code $\mathcal{C}$ can be seen to be equal to the bound given by \eqref{Thm1}.
\end{enumerate}
We now prove that $\mathcal{C}$ can correct $t$ erasures sequentially.
\vspace{-0.07cm}
\begin{enumerate}
	\item The Tanner graph of our code can be viewed as follows: We now view $G_{\frac{t}{2}-1}$ differently to give an alternative description of our code. In the graph $G_{\frac{t}{2}-1}$, let every edge represent a code symbol and every node represent a parity check of the code symbols corresponding to the edges incident on it. Apart from this, for every node in $U_{\frac{t}{2}-1}(\subseteq V(G_{\frac{t}{2}-1}))$ there is one distinct degree-1 variable node attached to it. Each of these degree-1 variable nodes corresponds to a distinct code symbol. The code symbol corresponding to a degree-1 variable node is part of the parity check represented by the node in $U_{\frac{t}{2}-1}$ to which it is attached. Our code is now defined by the code symbols represented by the edges of $G_{\frac{t}{2}-1}$ and degree-1 variable nodes attached to $U_{\frac{t}{2}-1}$. This completes the alternative description of our code.\\
	The alternative description of the code can be illustrated using the figures \ref{fig:Tanner1}, \ref{fig:Tanner2} and \ref{fig:Tanner3}.
	\begin{itemize}
		\item Consider a node $v \in T_i$ for $1 \leq i \leq \frac{t}{2}-2$ in $G_{\frac{t}{2}-1}$. In the original description of the code, the node $v$ represents a code symbol say $c_p$ which is the binary sum of $r$ symbols say $c_1,c_2,...,c_r$ where $c_1,c_2,...,c_r$ are represented by nodes in $N(v)\cap T_{i-1}$; in the alternative description of the code the same node $v$ now represents the parity check $c_p\oplus c_1\oplus c_2\oplus ...\oplus c_r=0$ and as shown in figure \ref{fig:Tanner1} the edge going out upward from the node $v$ represents the code symbol $c_p$ and the edges going out downward from $v$ represent the code symbols $c_1,c_2,...,c_r$. 
		\begin{figure}[ht]
			\centering
			\includegraphics[scale = 0.6]{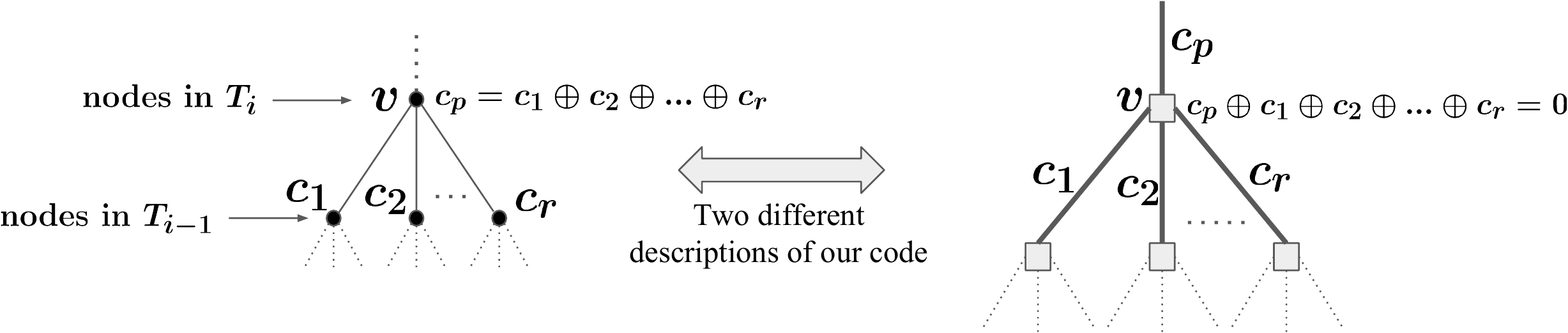}
			\caption{Figure showing two descriptions of our code at a node $T_i$ for $1 \leq i \leq \frac{t}{2}-2$}
			\label{fig:Tanner1}
		\end{figure}
		\item Consider a node $v \in U_{\frac{t}{2}-1}$ in $G_{\frac{t}{2}-1}$. In the original description of the code, the node $v$ represents a code symbol say $c_p$ which is the binary sum of $r$ symbols say $c_1,c_2,...,c_r$ where $c_1,c_2,...,c_r$ are represented by nodes in $N(v)\cap T_{\frac{t}{2}-2}$; in the alternative description of the code the same node $v$ now represents the parity check $c_p\oplus c_1\oplus c_2\oplus ...\oplus c_r=0$ and as shown in figure \ref{fig:Tanner2} the degree-1 node attached to $v$ represents the code symbol $c_p$ and the edges going out downward from $v$ represent the code symbols $c_1,c_2,...,c_r$. Hence there is a unique distinct degree-1 node attached to every node in $U_{\frac{t}{2}-1}$. These are exactly the degree-1 variable nodes we described above in the alternative description of the code. 		
		\begin{figure}[ht]
			\centering
			\includegraphics[scale = 0.65]{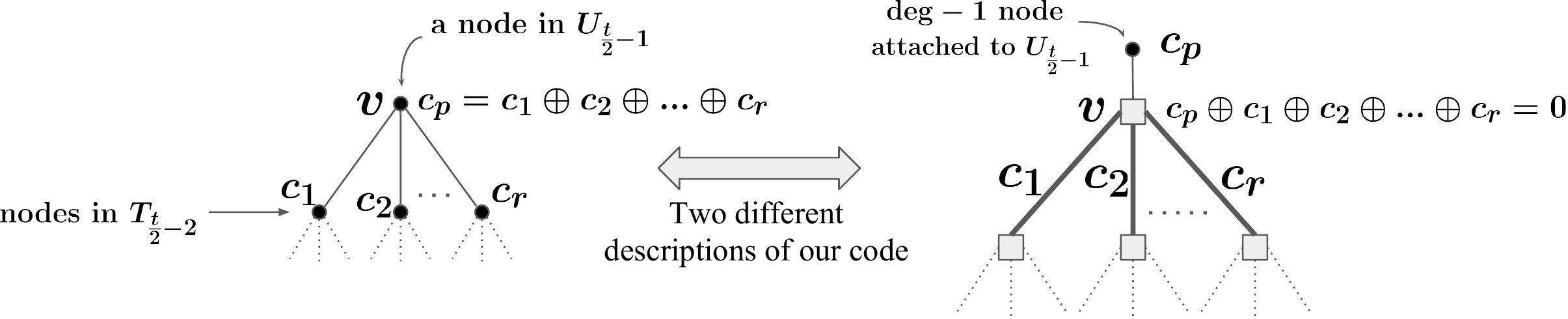}
			\caption{Figure showing two descriptions of our code at a node in $U_{\frac{t}{2}-1}$}
			\label{fig:Tanner2}
		\end{figure}
		\item Consider a node $v \in T_0$ in $G_{\frac{t}{2}-1}$. In the original description of the code, the node $v$ represents a code symbol say $c_p$ which is the binary sum of $r$ information symbols say $i_1,i_2,...,i_r$ where $i_1,i_2,...,i_r$ are represented by edges incident on $v$ in $G_0^{\text{rep}}$; in the alternative description of the code the same node $v$ now represents the parity check $c_p\oplus i_1\oplus i_2\oplus ...\oplus i_r=0$ and as shown in figure \ref{fig:Tanner3} the edge going out upward from $v$ represents the code symbol $c_p$ and the edges going out downward from $v$ still represents the information symbols $i_1,i_2,...,i_r$.
		\begin{figure}[ht]
			\centering
			\includegraphics[scale = 0.6]{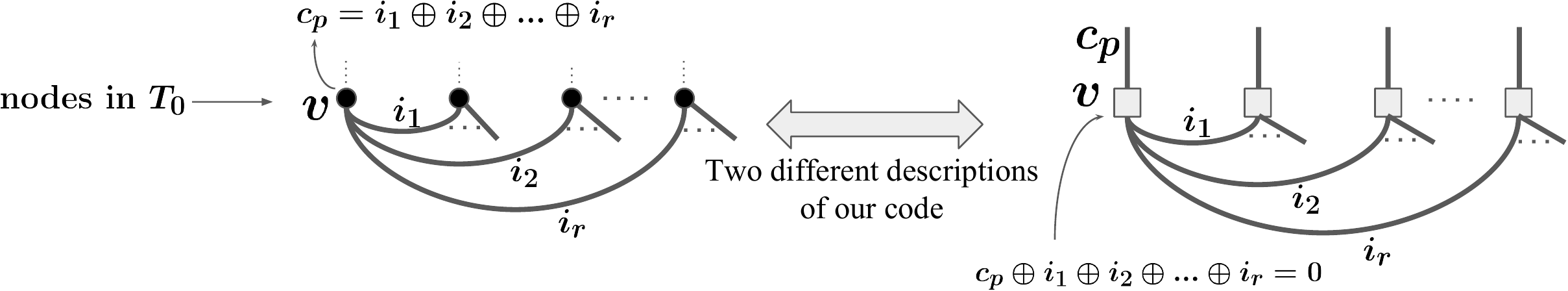}
			\caption{Figure showing two descriptions of our code at a node in $T_0$}
			\label{fig:Tanner3}
		\end{figure} 	
		\item We transform the meaning of nodes and edges in $G_{\frac{t}{2}-1}$ from original to alternative description of the code as described above across nodes and edges of the graph $G_{\frac{t}{2}-1}$ and it can be done maintaining consistency across the tranformations of meaning of nodes and edges from original to alternative description.
		\item By the above explanation, all the parity checks necessary to describe our code are represented by nodes in the graph $G_{\frac{t}{2}-1}$ in the alternative description.
	\end{itemize} 
	 Note that degree-1 variable nodes are not part of $G_{\frac{t}{2}-1}$ and these degree-1 variable nodes are the reason for rate optimality.
	\item Let $E$ be a set of erased symbols. Let the code symbols corresponding to degree-1 variable nodes be $D_1$. Let $E \cap D_1 = \emptyset$. Let an edge $(v_1,v_2)$ represent an erased symbol in $E$. Construct a maximal walk $P=v_1v_2...v_w$, $v_i \in V(G_{\frac{t}{2}-1})$ $,\forall i \in [w]$ such that $v_i \neq v_j$, $\forall\ w \geq i \neq j \geq 1$, and the edge $(v_i,v_{i+1})$ represents an erased code symbol in $E$ $\forall i \in [w-1]$ (walk $P$ is maximal in terms of number of vertices). Since the walk is maximal any edge of the form $(v_w,v)$ for $v \in N(v_w)-\{v_{w-1}\}$, $v \in V(G_{\frac{t}{2}-1})$:
	\begin{enumerate}
		\item either represents an unerased symbol for all $v \in N(v_w)-\{v_{w-1}\}$, $v \in V(G_{\frac{t}{2}-1})$ in which case the symbol represented by the edge $(v_{w-1},v_w)$ can be recovered by using the parity check represented by $v_w$.
		\item or represents an erased symbol and $v \in \{v_1,...,v_w\}$ in which case there is a cycle. Since girth of $G_{\frac{t}{2}-1}$ is $\geq t+1$, there must at least $t+1$ distinct erased symbols.
	\end{enumerate}
	Hence either an erased symbol can be recovered or there are more than $t$ erasures. The argument can be repeated for recovering subsequent erased symbols sequentially one by one if $|E|<t$.
	\item If $E \cap D_1 \neq \emptyset$, then start with a vertex $v_1 \in U_{\frac{t}{2}-1}$ to which a degree-1 variable node corresponding to an erased code symbol is attached. Now construct a maximal walk, $P=v_1...v_w$, $v_i \in V(G_{\frac{t}{2}-1}$),$\forall i \in [w]$, on the erased symbols as before. Now either as before, one of the symbols in $E$ can be recovered or there are more than $t$ erased symbols in $E$ due to a cycle or the walk ends in another node in $U_{\frac{t}{2}-1}$ i.e., $v_w \in U_{\frac{t}{2}-1}$ with the code symbol corresponding to degree-1 variable node attached to $v_w$ also being erased. Since $G_{\frac{t}{2}-1}$ is tree-like, any two nodes in $U_{\frac{t}{2}-1}$ are separated by a path of length at least $t-1$ in $G_{\frac{t}{2}-1}$. Hence there are at least $t-1+2=t+1$  erasures (+2 because $|E \cap D_1| \geq 2$).
	\item Thus the entire point of constructing  $G_{\frac{t}{2}-1}$ in a tree-like fashion is for protection from more than one erasure in the code symbols corresponding to degree-1 variable nodes attached to $U_{\frac{t}{2}-1}$.
\end{enumerate}
\subsection{Construction of codes achieving the upper bound on rate for $t$ odd:}
In this subsection we give a construction of codes achieving the rate bound \eqref{Thm2} for any $r \geq 3$ and $t = 2s-1, s \geq 1$.
A construction achieving the bound \eqref{Thm2} for $t=5$ can be found in \cite{balaji2016binary}. For $t=3$, the bound \eqref{Thm2} can be achieved by taking the product of two $[r+1,r]$ single parity check code (2 dimensional product code). For $t=1$, the bound \eqref{Thm2} can be achieved by taking the $[r+1,r]$ single parity check code. Hence we give a rate optimal construction for any $r \geq 3$ and $t = 2s-1$, $s \geq 2$ in the following.\\
We provide an iterative, graph-based construction of rate-optimal codes for a given $r \geq 3$ and $t=2s-1$, $s \geq 2$. We build a new graph $G_{s-1}$ iteratively, starting from $G_0$ (different from the one considered for $t$ even case.), by adding nodes in vertical, layer-by-layer fashion.  At every stage of the iteration, the graph $G_i$ (again different from the one constructed for $t$ even case.) thus constructed, will always have girth $\geq t+1$.  We then define our rate-optimal code based on the graph $G_{s-1}$.\\

\textbf{Construction of graph $G_{s-1}$:}
Denote the vertex set of a graph $G$ by $V(G)$. We make use of graphs $G_0$ and $B_i$, for $1 \leq i \leq s-1$ described below as the ingredients:
\ben
\item We pick $G_0$ as follows: Let $G_0$ be an $r$-regular bipartite graph with girth $\geq t+1$\textsuperscript{\ref{Fur}}.  Let $U_{0,1},U_{0,2}$ be the two (upper and lower) sets of nodes in the bipartite graph with $\text{degree}(x)=r,\forall\ x\in U_{0,1}$ and $\text{degree}(y)=r,\forall\ y\in U_{0,2}$. Hence $V(G_0)=U_{0,1} \cup U_{0,2}$. Define $|U_{0,1}|=|U_{0,2}|=u_0$.  
%The construction requires the base-graph $G_0$ to be $r$-regular with girth $\geq t+1$ and number of nodes that is a multiple of $2r^{s-1}$. Hence we take the disjoint union of a number of 
%%$\frac{\text{LCM}(|V(G)|,r^{\frac{t}{2}-1})}{|V(G)|}$ 
%copies of $G$, with the number chosen to ensure that the size of the resultant graph is a multiple of $2r^{\frac{t}{2}-1}$.  We then declare the resultant graph to be $G_0$.
\item Next, for $i$ in the range $1 \leq i \leq s-1$, we iteratively construct an $(r,2u_{i-1})$-biregular bipartite graph $B_i$ with girth $g_i \geq \left \lceil \frac{t+1}{i+\frac{1}{2}} \right \rceil$ with the following properties. Let $U_i,L_i$ be the two (upper and lower) sets of nodes in the bipartite graph with $\text{degree}(x)=r,\forall\ x\in U_i$ and $\text{degree}(y)=2u_{i-1},\forall\ y\in L_i$. Hence $V(B_i)=U_i \cup L_i$. We have that  $U_i = U_{i,1} \cup U_{i,2}$ with $|U_{i,1}|=|U_{i,2}|$ and $U_{i,1} \cap U_{i,2} = \emptyset$. Let $|U_{i,1}|=|U_{i,2}|=u_i$ and $|L_i|=l_i$. Therefore $|U_i|=2u_i$. By edge counting, we have $2ru_i = 2u_{i-1}l_i$ i.e. $ru_i = u_{i-1}l_i$. Additionally we have that, $\forall \ x \in L_i, |N(x)\cap U_{i,1}|=|N(x)\cap U_{i,2}|=u_{i-1}$ where $N(x)$ denotes the neighbors of node $x$ in $B_i$. Construction of such graphs will be described separately. 
\een

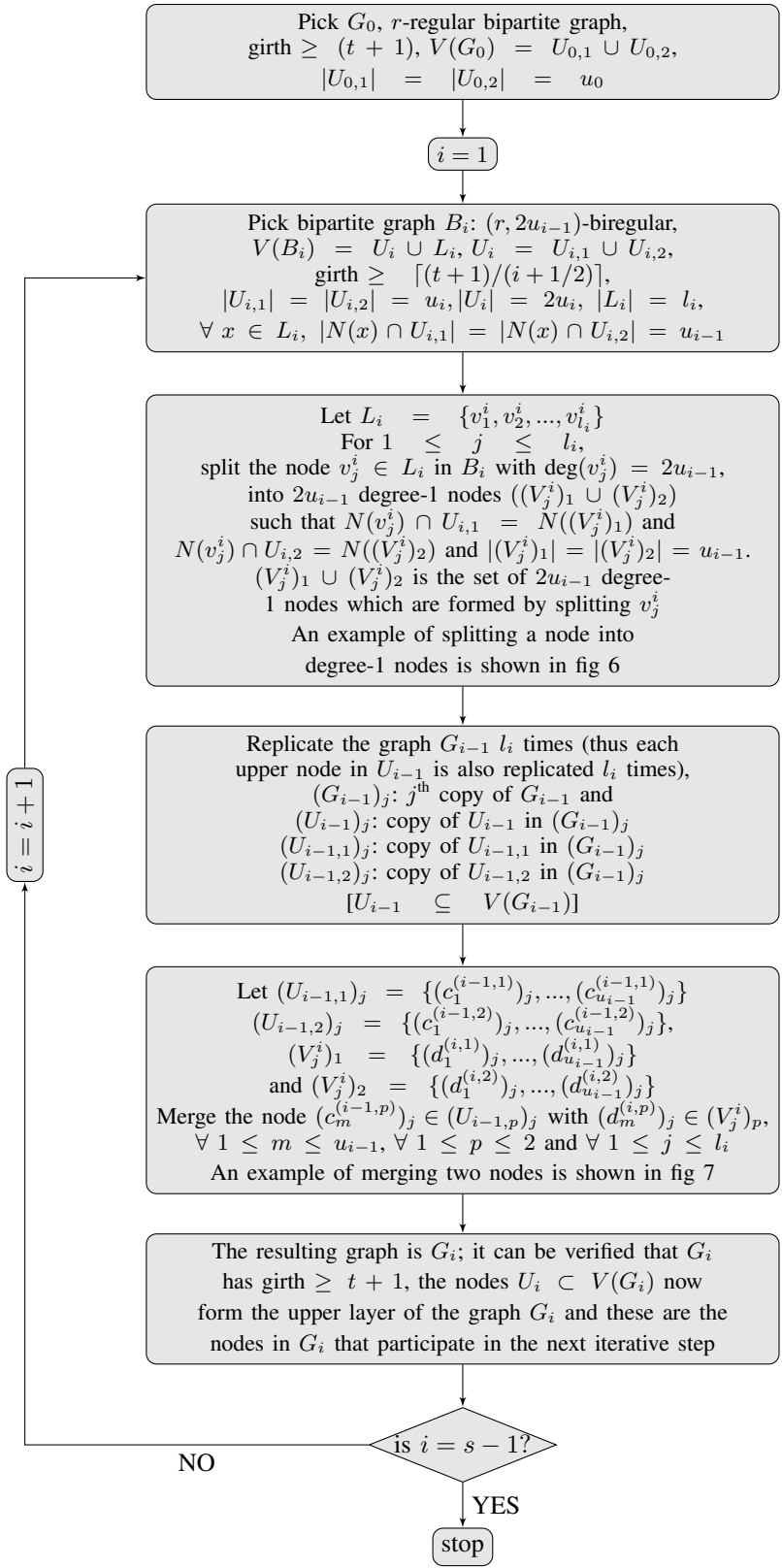
\begin{figure}
	\centering 
	\tikzstyle{decision} = [diamond, draw, fill=gray!20, 
	text badly centered, node distance=2cm, inner sep=0pt, aspect=2.5]
	
	\tikzstyle{block} = [rectangle, draw, fill=gray!20, 
	text width=24em, text centered, rounded corners, node distance=2.4cm]
	
	\tikzstyle{block2} = [rectangle, draw, fill=gray!20, 
	text centered, rounded corners, node distance=1.4cm]
	
	\tikzstyle{block3} = [rectangle, draw, fill=gray!20,  text centered, rounded corners, node distance=1.4cm]    
	
	\tikzstyle{line} = [draw, -latex']
	
	\tikzstyle{cloud} = [draw, ellipse,fill=red!20, node distance=2cm]
	
	\begin{tikzpicture}[node distance = 1.8cm, auto]
	% Place nodes
	\node [block] (INIT) {\small{Pick $G_0$, $r$-regular bipartite graph,\\
    girth $\geq(t+1)$, $V(G_0)=U_{0,1} \cup U_{0,2}$, \\$|U_{0,1}|=|U_{0,2}|=u_0$}};
	
	\node [block3, below of=INIT, yshift=0cm] (INITIALI) {\small{$i=1$}} ;
	
	\node [block, below of=INITIALI, yshift=0.7cm] (BIPARTITE) {\small{Pick bipartite graph $B_i$: $(r,2u_{i-1})$-biregular,  \\ $V(B_i)=U_i \cup L_i$, $U_i = U_{i,1} \cup U_{i,2}$,\\ girth $\geq  \left \lceil (t+1)/(i+1/2) \right \rceil$, \\ $|U_{i,1}|=|U_{i,2}|=u_i,|U_i|=2u_i,\  |L_i|=l_i$,\\ $\forall \ x \in L_i,\ |N(x)\cap U_{i,1}|=|N(x)\cap U_{i,2}|=u_{i-1}$}};
	
	\node [block, below of=BIPARTITE, yshift=-1.2cm] (SPLIT) {\small{Let $L_i=\{v_1^i,v_2^i,...,v_{l_i}^i\}$ \\ For $1 \leq j \leq l_i$, \\ split the node $v_j^i \in L_i$ in $B_i$ with $\text{deg}(v_j^i)=2u_{i-1}$, into $2u_{i-1}$ degree-1 nodes $((V_j^i)_1 \cup (V_j^i)_2)$} such that $N(v_j^i) \cap U_{i,1}=N((V_j^i)_1)$ and $N(v_j^i) \cap U_{i,2}=N((V_j^i)_2)$ and $|(V_j^i)_1| =|(V_j^i)_2| = u_{i-1}$. \\ $(V_j^i)_1 \cup (V_j^i)_2$ is the set of $2u_{i-1}$ degree-1 nodes which are formed by splitting $v_j^i $ \\An example of splitting a node into degree-1 nodes is shown in fig \ref{fig:Splitting}};
	
	\node [block, below of=SPLIT, yshift=-1.5cm] (REPLICATE) {\small{Replicate the graph $G_{i-1}$ $l_i$ times (thus each upper node in $U_{i-1}$ is also replicated $l_i$ times), \\ $(G_{i-1})_j$: $j^{\text{th}}$ copy of $G_{i-1}$ and $(U_{i-1})_j$: copy of $U_{i-1}$ in $(G_{i-1})_j$ \\
			$(U_{i-1,1})_j$: copy of $U_{i-1,1}$ in $(G_{i-1})_j$\\
			$(U_{i-1,2})_j$: copy of $U_{i-1,2}$ in $(G_{i-1})_j$\\ $\text{[}U_{i-1} \subseteq V(G_{i-1})\text{]}$}};
	
	\node [block2,  left of=REPLICATE, node distance=4cm, xshift=-2cm] (UPDATE) {\rotatebox[]{90}{$i=i+1$}};
	
	\node [block, below of=REPLICATE, yshift=-1.1cm] (MERGE) {\small{Let $(U_{i-1,1})_j=\{(c_{1}^{(i-1,1)})_j,...,(c_{u_{i-1}}^{(i-1,1)})_j\}$ \\
	$(U_{i-1,2})_j=\{(c_{1}^{(i-1,2)})_j,...,(c_{u_{i-1}}^{(i-1,2)})_j\}$,\\
	$(V_j^i)_1=\{(d_{1}^{(i,1)})_j,...,(d_{u_{i-1}}^{(i,1)})_j \}$\\ and $(V_j^i)_2=\{(d_{1}^{(i,2)})_j,...,(d_{u_{i-1}}^{(i,2)})_j \}$ \\ Merge the node $(c_{m}^{(i-1,p)})_j \in (U_{i-1,p})_j$ with $(d_{m}^{(i,p)})_j \in (V_j^i)_p$, \\ $\forall\ 1 \leq m \leq u_{i-1}$, $\forall\ 1 \leq p \leq 2$ and $\forall\ 1 \leq j \leq l_i$ \\An example of merging two nodes is shown in fig \ref{fig:Merging} }};
	
	\node[block, below of=MERGE, yshift=-0.6cm] (GIGRAPH) {\small{The resulting graph is $G_i$; it can be verified that $G_i$ has girth $\geq t+1$, the nodes $U_i \subset V(G_i)$ now form the upper layer of the graph $G_i$ and these are the nodes in $G_i$ that participate in the next iterative step}};
	
	\node [decision, below of=GIGRAPH] (DECIDE) {is $i = s-1?$};
	
	\node [block3, below of=DECIDE] (stop) {stop};
	
	% Draw edges
	\path [line] (INIT) -- (INITIALI);
	\path [line] (INITIALI) -- (BIPARTITE);
	\path [line] (BIPARTITE) -- (SPLIT);
	\path [line] (SPLIT) -- (REPLICATE);
	\path [line] (REPLICATE) -- (MERGE);
	\path [line] (MERGE) -- (GIGRAPH);
	\path [line] (GIGRAPH) -- (DECIDE);
	\path [line] (DECIDE) -| node [near start] {NO} (UPDATE);
	\path [line] (UPDATE) |- (BIPARTITE);
	\path [line] (DECIDE) -- node {YES}(stop);
	%   \path [line,dashed] (expert) -- (init);
	%   \path [line,dashed] (system) -- (init);
	%   \path [line,dashed] (system) |- (replicatebi);
	
	\end{tikzpicture}
	
	\caption{Flowchart showing the iterative construction of $G_{s-1}$.}
	\label{fig:FlowchartOdd}
\end{figure}

   \textbf{Construction of $B_i$}:
      \begin{itemize}
      	\item Consider a $(r,u_{i-1})$-biregular bipartite graph $G^1 = (X_1 \cup X_2,E)$ with girth at least $\left\lceil\frac{t+1}{i+\frac{1}{2}}\right\rceil$. Let $X_1 = \{b_1,...,b_{\hat{m}}\}$, $X_2 = \{c_1,...,c_{\lambda}\}$ for some $\hat{m},\lambda$ with $\text{degree}(b_j) = r$, for $1 \leq j \leq \hat{m}$ and $\text{degree}(c_j) = u_{i-1}$, for $1 \leq j \leq \lambda$. Such a graph $G^1$ can be constructed due to \cite{Furedi}.
      	\item Consider a $\lambda$-regular bipartite graph $G^2=(X \cup Y,E)$ with  girth at least $\frac{1}{2}\left\lceil\frac{t+1}{i+\frac{1}{2}}\right\rceil$. Let $X = \{x_1,...,x_{l'}\}$ and $Y = \{y_1,...,y_{l'}\}$ for some $l'$ with each node of degree $\lambda$. Such a graph $G^2$ can be constructed due to \cite{Furedi}.
      	\item Take $2l'$ disjoint copies of $G^1$. Let $G^3$ be the disjoint union of the $2l'$ copies named $G^1_1,...,G^1_{2l'}$ of $G^1$. Denote the nodes $b_j,c_j$ in the $l^{th}$ copy $G^1_l$ by $b^l_j,c^l_j$ respectively. 
      	\item Replace the node $x_f$ in $G^2$ by the nodes $x_{1,f},...,x_{\lambda,f}$ for $1 \leq f \leq l'$. If the neighbors of $x_f$ in $G^2$ are the nodes in the set $\{y_{q_1},...,y_{q_\lambda}\}$ then connect $x_{j,f}$ to $y_{q_j}$ for $1 \leq j \leq \lambda$. Now replace the node $y_f$ by the nodes $y_{1,f},...,y_{\lambda,f}$ for $1 \leq f \leq l'$. If the neighbors of $y_f$ after replacing the nodes in $X$ were the nodes in the set $\{x_{\sigma_1,\beta_1},...,x_{\sigma_\lambda,\beta_\lambda}\}$ then connect $y_{j,f}$ to $x_{\sigma_j,\beta_j}$ for $1 \leq j \leq \lambda$. Denote this graph by $G^4$. For every edge $(x_{j,f},y_{\hat{j},\hat{f}})$ in $G^4$ merge the nodes $c^f_j$ and $c^{l'+\hat{f}}_{\hat{j}}$ in $G^3$ and label the resulting merged node as $v^i_{(f-1)\lambda+j}$ in $G^3$. Relabel $b^l_j$ as $c^{(i,1)}_{(l-1)\hat{m}+j}$, $\forall 1 \leq l \leq l'$,$\forall 1 \leq j \leq \hat{m}$ in $G^3$. Relabel $b^l_j$ as $c^{(i,2)}_{(l-l'-1)\hat{m}+j}$, $\forall \ l'+1 \leq l \leq 2l'$,$\forall 1 \leq j \leq \hat{m}$ in $G^3$ . After the merging and relabeling of nodes in $G^3$ as described, resulting graph is named as $B_i$.
      	\item The constructed graph $B_i$ is a $(r,2u_{i-1})$-biregular bipartite graph $B_i=(U_i \cup L_i, E_i)$ with $U_i=U_{i,1} \cup U_{i,2}$ where $U_{i,1}=\{c_{1}^{(i,1)},...,c_{u_{i}}^{(i,1)}\}$, $U_{i,2}=\{c_{1}^{(i,2)},...,c_{u_{i}}^{(i,2)}\}$ and $u_i = l'\hat{m}$ and $L_i=\{v^i_1,...,v^i_{l_i}\}$ where $l_i = l'\lambda$ and $\text{degree}(c_j^{(i,p)})=r$ for $1 \leq j \leq u_i$,$1 \leq p \leq 2$ and $\text{degree}(v^i_j)=2u_{i-1}$ for $1 \leq j \leq l_i$. It can be seen that $B_i$ has girth at least $\left\lceil \frac{t+1}{i+\frac{1}{2}}\right\rceil$. It can also be seen from the construction of $B_i$ that $|N(v^i_j) \cap U_{i,1}|=|N(v^i_j) \cap U_{i,2}|=u_{i-1}$, for $1 \leq j \leq l_i$.
      \end{itemize}
      \textbf{Construction of $G_{s-1}$:}
      The flowchart shown in Fig.~\ref{fig:FlowchartOdd} describes the iterative construction of the graph $G_{s-1}$. 
      
      We make the following observations (similar observations compared to $t$ even case):
      \begin{enumerate}
      	\item For $0 \leq i \leq s-1$, there are $f_i=\prod_{j=i+1}^{s-1}l_j$ disjoint copies of $G_i$ that reside inside $G_{s-1}$ as subgraphs due to the replication steps in the construction of $G_{s-1}$ (given in the flowchart Fig \ref{fig:FlowchartOdd}). For all $s-1 > j>i$, since copies of $G_i$ reside inside $G_j$ as subgraphs, when we replicate $G_j$ in the $(j+1)^{th}$ replication step, all the copies of $G_i$ contained within $G_j$ also get replicated; we refer to all such copies of $G_i$ here. Let us denote the disjoint union of these subgraphs corresponding to all copies of $G_i$ formed due to replication steps by $G_i^{rep}=\cup_{j=1}^{f_i} (G_i)_j$, where $(G_i)_j$ is the $j^{th}$ copy of $G_i$ in $G_{s-1}$ formed due to some replication step. We view $G_i^{rep}$ as a subgraph of $G_{s-1}$.
      	\item Note that the construction of $G_{s-1}$ proceeds by adding nodes layer by layer with each layer connecting to the layer below it in a tree-like fashion while maintaining the girth of overall graph to be atleast $t+1$. Let $N(v)$ represent neighbors of a node $v$ in $G_{s-1}$. Let $(U_i)_j$ be the copy of $U_i$ in $V((G_i)_j), \forall 1 \leq j \leq f_i$. For $1 \leq i \leq s-1 $, the nodes $T_i=\cup_{j=1}^{f_i} (U_i)_j$ represent the nodes in the $i^{th}$ layer and they connect in a tree-like fashion to the $(i-1)^{th}$ layer nodes $T_{i-1}=\cup_{j=1}^{f_{i-1}} (U_{i-1})_j$. The connection is tree-like because for $v \in T_i$, let $S_v = N(v) \cap T_{i-1}$ then $S_v \cap S_w = \emptyset$ and $|S_v|=|S_w|=r$ $,\forall v,w \in T_i,  v \neq w$  and $N(v) \cap T_p = \emptyset ,\forall p \notin \{i-1,i+1\}$,$\forall v \in T_i$. 
      	\item For $1 \leq p \leq 2$, let $(U_{i,p})_j$ be the copy of $U_{i,p}$ in $V((G_i)_j), \forall 1 \leq j \leq f_i$. In fact it can be seen that for $1 \leq i \leq s-1$, the nodes in $T^p_{i}  = \cup_{j=1}^{f_i} (U_{i,p})_j$ connect in a tree-like fashion to nodes in $T^p_{i-1}  = \cup_{j=1}^{f_i} (U_{i-1,p})_j$ for $p \in \{1,2\}$. The connection is tree-like because for $v \in T^p_i$, let $S^p_v = N(v) \cap T^p_{i-1}$ then $S^p_v \cap S^p_w = \emptyset$ and $|S^p_v|=|S^p_w|=r$ $,\forall v,w \in T^p_i,  v \neq w$  and $N(v) \cap T^p_q = \emptyset ,\forall q \notin \{i-1,i+1\}$,$\forall v \in T^p_i$.  This implies a tree-like graph is built starting from bottom layer $T^p_0$ (leaf-like nodes) towards $T^p_{s-1}$. It can seen that the tree-like graph built from $T^1_0$ does not intersect or disjoint with tree-like graph built from $T^2_0$. The only interaction between these 2 tree-like graphs is by edges representing information symbols connecting nodes in $T^1_0$ with $T^2_0$.
      	\item The fact that $G_i$ has girth at least $t+1$ can be proved using arguments similar to those used in $t$ even case. Hence we skip the proof.
      \end{enumerate}
      \textbf{Description of a rate-optimal code $\mathcal{C}$ based on $G_{s-1}$:}\\
      Code $\mathcal{C}$ is defined on the graph $G_{s-1}$. We will first describe a code $\mathcal{C}_1$ based on $G_{s-1}$. On puncturing the code $\mathcal{C}_1$ by dropping few parity symbols, we get our desired rate-optimal code $\mathcal{C}$.\\
      Code $\mathcal{C}_1$ is defined on the graph $G_{s-1}$ as follows:
      \begin{enumerate}
      	\item The edges of $G_0^{rep}$ represent information symbols. It can be seen that the number of information symbols is $k = u_0r\prod_{j=1}^{s-1}l_j$.
      	\item Every node of $G_{s-1}$ represents a distinct parity symbol described as follows.
      	\item A node $v \in V(G_0^{rep})$ represents a parity symbol which is the binary sum of information symbols that are represented by edges in $G_0^{rep}$ incident on $v$. 
      	\item For $1 \leq i \leq s-1$, a node $v \in T_i$ in $G_{s-1}$ represents a parity symbol which is the binary sum of code symbols that are represented by the nodes in $N(v) \cap T_{i-1}$.
      	
      	%	every node in the sets $U_{i}$ of the copies of $G_{i}$ in the graph $G_i^{rep}$ represents a parity symbol which is the binary sum of symbols represented by its neighboring nodes in $G_i^{rep}$.\\
      	%	 Equivalently, a parity symbol represented by a node $v \notin V(G_0^{rep})$ in $G_{\frac{t}{2}-1}$ is the binary sum of code symbols represented by neighbors of $v$ which are introduced before $v$ in the iterative construction of $G_{\frac{t}{2}-1}$ represented by the flow chart Fig \ref{fig:Flowchart}.
      	\item It can be seen that the total number of nodes is $n'-k = \sum_{i=0}^{s-1}2u_i\prod_{j=i+1}^{s-1}l_j = 2\sum_{i=0}^{s-1}\frac{u_0}{r^i}\prod_{j=1}^{s-1}l_j$.
      	\item $\mathcal{C}_1$ is an $[n',k]$ code defined by the information symbols represented by edges in  $G_0^{rep}$ and parity symbols represented by nodes in $G_{s-1}$.
      \end{enumerate}
      Code $\mathcal{C}$ is defined as follows:
      \begin{enumerate}
      	\item Puncture the code $\mathcal{C}_1$ by dropping the parity symbols represented by nodes in $U_{s-1,2}$. Hence $\mathcal{C}$ is an $[n,k]$ code defined by the information symbols represented by edges in  $G_0^{rep}$ and parity symbols represented by nodes in $V(G_{s-1})-U_{s-1,2}$.
      	\item Hence $n-k = n'-k-u_{s-1} = 2\sum_{i=0}^{s-2}u_i\prod_{j=i+1}^{s-1}l_j + u_{s-1}= 2 \sum_{i=0}^{s-2}\frac{u_0}{r^i}\prod_{j=1}^{s-1}l_j + \frac{u_0}{r^{s-1}}\prod_{j=1}^{s-1}l_j$.
      	\item From the above counts, the rate of the code $\mathcal{C}$ can be seen to be equal to the bound given by \eqref{Thm2}.
      \end{enumerate}
       The proof that $\mathcal{C}$ can correct $t$ erasures sequentially is similar to the $t$ even case. Hence we skip the proof but we give the main idea as follows. The only point to be noted for sequential recovery of $t$ erasures considering the cases we have dealt with in $t$ even case is that nodes in the top most layer which are part of the code $\mathcal{C}$ i.e., nodes in $T^1_{s-1} \cup T^2_{s-2}$ are separated from each other with just enough distance that a path of erased symbols connecting any two nodes (representing erased symbols) at this top-most layer (similar to $t$ even case) must have at least $t+1$ erased symbols. Since the path between 2 nodes in  $T^2_{s-2}$ is potentially smaller than the the path between two nodes in  $T^1_{s-1}$, we need a little extra distance separating two nodes in $T^2_{s-2}$ which is given by the bipartite nature of the graph at bottom most layer since the path from a node in $T^2_{s-2}$ can reach another node in $T^2_{s-2}$ only through two nodes $T^2_0$ we get 1 extra distance as any two nodes in $T^2_0$ is separated by a distance of atleast 2 due to the bipartite nature of $G_0^{rep}$. This is precisely the reason why we built tree-like graph from $T^1_0$ disjointly from $T^2_0$, so that when we puncture which drops the distance between nodes in top most layer, we get a little extra distance at the bottom layer.
\begin{note}
	The following conjecture on the rate of an $(n,k,r,t)_{seq}$ code appeared in \cite{song2016sequential}.
	\bean
	\frac{k}{n} \leq \frac{1}{1+\sum_{i=1}^{m}\frac{a_i}{r^i}}, \\
	a_i \geq 0,a_i \in \mathbb{Z}, \sum_{i=1}^{m} a_i = t, \\
	m=\lceil log_r(k) \rceil.
	\eean
	
	Our rate bound verifies the conjecture as our rate bound can be written in the form:\\
	
	For $t$ even:
	\bean
	\frac{k}{n} \leq \frac{1}{1+\sum_{i=1}^{\frac{t}{2}} \frac{2}{r^i}}.
	\eean

	For $t$ Odd:
	\bean
	\frac{k}{n} \leq \frac{1}{1+\sum_{i=1}^{s-1} \frac{2}{r^i} + \frac{1}{r^s}}.
	\eean
	
	More specific conjecture was given for $t=5,6$. While our bound proves the conjecture for $t=5$, the conjectured upper bound for $t=6$ does not hold as our upper bound on code rate is both larger and achievable).\\
	
	For $t=6$, Our bound takes the form:
	\bean
	\frac{k}{n} \leq \frac{r^3}{r^3+2r^2+2r+2} \\
	\frac{k}{n} \leq \frac{1}{1+\frac{2}{r}+\frac{2}{r^2}+\frac{2}{r^3}} \\
	\eean
	The conjecture for $t=6$ given in \cite{song2016sequential}:
	\bean
	\frac{k}{n} \leq \frac{1}{1+\frac{2}{r}+\frac{3}{r^2}+\frac{1}{r^3}} \\
	\eean
	
	Hence the conjecture given in \cite{song2016sequential} has a smaller rate than our bound for $t=6$. And further our upper bound on rate for $t=6$ can be achieved.
\end{note}

\begin{note}
	In \cite{RawMazVis} the authors provide a construction of a code with sequential recovery for any $r$ and $t$ having rate $\frac{r-1}{r+1}+\frac{1}{n}$. $\frac{r-1}{r+1}+\frac{1}{n}$ will meet our rate bound exactly when a Moore graph of degree $r+1$ and girth $t+1$ exists. Moore graphs(degree = $r+1$, girth = $t+1$) are shown to not exist for any $t \notin \{2,3,4,5,7,11\}$ for any $r \geq 2$ (see \cite{DynCageSur}). Hence the construction in \cite{RawMazVis} is not rate-optimal for most of the cases.
\end{note}
	
\appendices
\section{Proof of Claim \ref{claim1}} \label{AppendixA}
It is enough to show that:

\begin{itemize}
	\item $\{A_i\}$ are matrices with each column having Hamming weight 1.
	\item $\{D_i\}$ are matrices with each row having Hamming weight 1. 
\end{itemize}
As the point 1 written above combined with the fact that each column of $\left[\frac{A_i}{D_i}\right]$ has Hamming weight 2  implies $D_i,i \geq 1$ are matrices with each column having Hamming weight 1 and $D_0$ by definition is a matrix with each column having Hamming weight 1. \\
Let us denote the $j^{\text{th}}$ column of matrix $H_1$ by $\underline{h}_j$, $1 \leq j\leq n$. Let us show the claim by induction as follows:\\
(Note that in the arguments that follow $\underline{h}_{p_1},...,\underline{h}_{p_\psi}$, $\underline{h}_{y_1},...,\underline{h}_{y_\phi}$,  $\underline{h}_{z_1},...,\underline{h}_{z_\theta}$ and $\underline{h}_{w_1},...,\underline{h}_{w_\gamma}$ will be repeatedly used each time denoting some set of column vectors not necessarily related to a different instance of the same notation.)\\

Induction Hypothesis: 
\begin{itemize}
	\item  Property $P_i$: any $m \times 1$ vector having Hamming weight at most 2 with support contained in $\{\sum_{l=0}^{i-1}\rho_l+1,..,\sum_{l=0}^{i-1}\rho_l+\rho_i\}$  can be written as some linear combination of at most $2(i+1)$ column vectors of $H_1$ say $\underline{h}_{p_1},...,\underline{h}_{p_\psi}$ for some $\{p_1,...,p_\psi\} \subseteq \{1,...,\sum_{l=0}^{i}a_l\}$ and $0 < \psi \leq 2(i+1)$.
	\item Let us assume as induction hypthesis that the property $P_i$ is true and the Claim \ref{claim1} is true for $A_1,...A_i$, $D_0,...D_i$.
\end{itemize}
Initial step:
\begin{itemize}	
	\item We show that each row of $D_0$ has Hamming weight exactly 1.\\
	Suppose there exists a row of $D_0$ with Hamming weight more than 1; let the support set of the row be $\{i_1,i_2,...\}$. Then the columns $\underline{h}_{i_1},\underline{h}_{i_2}$ of $H_1$ can be linearly combined to give a zero column. This contradicts the fact that $d_{min}({\mathcal{C}}) \geq t+1, t > 0$ and $t$ is even.
	Hence, all rows of $D_0$ have Hamming weight exactly 1.
	\item If $t = 2$, then the claim is already proved. So let $t \geq 4$.		
	\item We show that each column of $A_1$ has Hamming weight exactly 1.\\
	Suppose $j^{th}$ column of $A_1$ $\text{for some } 1 \leq j \leq a_1$ has Hamming weight 2; let the support of the column be $\{j_1,j_2\}$ in $A_1$. Then the column $\underline{h}_{a_0 + j}$ in $H_1$ along with the 2 column vectors of $H_1$ say $\underline{h}_{p_1},\underline{h}_{p_2}$ where $\{p_1,p_2\} \subseteq \{1,...,a_0\}$ where $\underline{h}_{p_1}$ has exactly one non-zero entry in  $j_1^{th}$ co-ordinate and $\underline{h}_{p_2}$ has exactly one non-zero entry $j_2^{th}$ co-ordinate, can be linearly combined to give a zero column again leading to a contradiction on minimum distance. Such columns with one column having only one non-zero entry exactly in $j_1^{th}$ co-ordinate and another column having only one non-zero entry exactly in $j_2^{th}$ co-ordinate with column labels in $\{1,...,a_0\}$ exist due to the single weight columns in the matrix $D_0$.
	\item The above argument also shows that any $m \times 1$ vector having Hamming weight at most 2 with support contained in $\{1,..,\rho_0\}$ can be written as some linear combination of at most 2 column vectors of $H_1$ say $\underline{h}_{p_1},...,\underline{h}_{p_\psi}$ for some $\{p_1,...,p_\psi\} \subseteq \{1,...,a_0\}$ ($\psi = 1 \text{ or }2$). Hence Property $P_0$ is true.		
	\item We now show that each row of $D_1$ has Hamming weight exactly 1.
	Suppose $j^{th}$ row of $D_1$ has Hamming weight more than 1; let the support set of the row be $\{l_1,l_2,...,l_z\}$ in $D_1$. Now there is some linear combination of columns $\underline{h}_{a_0+l_1}$ and $\underline{h}_{a_0+l_2}$ in $H_1$ that gives a zero in $(\rho_0+j)^{th}$ co-ordinate and thus this linear combination has support contained in $\{1,...,\rho_0\}$ with Hamming weight at most 2. Now applying Property $P_0$ on this linear combination implies that there is a non-empty set of at most $4$ linearly dependent columns in $H_1$ leading to a contradiction on minimum distance.		
	\item Now we show that Property $P_{1}$ is true: We have to prove that any $m \times 1$ vector with Hamming weight at most $2$ with support contained in $\{\rho_0+1,..,\rho_0+\rho_1\}$ can be written as linear combination of at most $2(1+1)=4$ column vectors of $H_1$ say $\underline{h}_{p_1},...,\underline{h}_{p_\psi}$ for some $\{p_1,...,p_\psi\} \subseteq \{1,...,\sum_{l=0}^{1}a_l\}$ and $0 < \psi \leq 4$. This can be easily seen using arguments similar to ones presented before.
	Let an $m \times 1$ vector have non-zero entries exactly in co-ordinates $\rho_0+j_1,\rho_0+j_2$ or $\rho_0+j_1$. Then this vector can be linearly combined with at most 2 column vectors in $H_1$ say $\underline{h}_{y_1},...,\underline{h}_{y_\phi}$ where $\{y_1,...,y_\phi\} \subseteq \{ a_0+1,...,a_0+a_{1} \}$ ($\phi = 1 \text{ or }2$) (2 columns $\underline{h}_{y_1},\underline{h}_{y_2}$ with first and second column having a non-zero entry in co-ordinates $\rho_0+j_1,\rho_0+j_2$ respectively or a column $\underline{h}_{y_1}$ with a non-zero entry in $(\rho_0+j_1)^{th}$ co-ordinate. These columns $\underline{h}_{y_1},\underline{h}_{y_2}$ or $\underline{h}_{y_1}$ exist due to $D_1$.) to form a $m \times 1$ vector with Hamming weight at most 2 with support contained in $\{ 1,...,\rho_0 \}$ which in turn can be written as linear combination of at most $2$ column vectors in $H_1$ say $\underline{h}_{z_1},...,\underline{h}_{z_\theta}$ for some $\{z_1,...,z_\theta\} \subseteq \{1,...,a_0\}$ ($\theta = 1$ or $2$) by property $P_0$. Hence the given $m \times 1$ vector is written as linear combination of at most $2(1+1)=4$ column vectors in $H_1$ say $\underline{h}_{w_1},...,\underline{h}_{w_\gamma}$ for some $\{w_1,...,w_\gamma\} \subseteq \{1,...,\sum_{l=0}^{1}a_l\}$ and $0 < \gamma \leq 4$.
	
\end{itemize}
Induction step :
\begin{itemize}
	\item Let us assume by induction hypothesis that Property $P_i$ is true and the Claim \ref{claim1} is true for $A_1,...A_i$, $D_0,...D_i$ for some $i \leq \frac{t}{2}-2$ and prove the induction hypothesis for $i+1$. For $t=4$, the initial step of induction completes the proof of Claim \ref{claim1}. Hence assume $t>4$.
	\item Now we show that each column of $A_{i+1}$ has Hamming weight exactly 1: suppose $j^{th}$ column of $A_{i+1}$ for some $1 \leq j \leq a_{i+1}$ has Hamming weight 2; let the support of the column be $j_1,j_2$ in $A_{i+1}$. It is clear that the corresponding column vector  $\underline{h}_{\sum\limits_{l=0}^{i}a_l+j}$ in $H_1$ is a vector with support contained in $\{\sum_{l=0}^{i-1}\rho_l+1,..,\sum_{l=0}^{i-1}\rho_l+\rho_i\}$ and Hamming weight 2. Now applying Property $P_i$ on this column vector $\underline{h}_{\sum\limits_{l=0}^{i}a_l+j}$ implies that there is a non-empty set of at most $2(i+1)+1$ columns in $H_1$ which are linearly dependent; hence contradicts the minimum distance as $2(i+1)+1 \leq t-1$. Hence each column of $A_{i+1}$ has Hamming weight exactly 1.
	\item Now we show that each row of $D_{i+1}$ has Hamming weight exactly 1: suppose $j^{th}$ row of $D_{i+1}$ has Hamming weight more than 1; let the support set of the row be $\{l_1,...,l_z\}$ in $D_{i+1}$. Now some linear combination of columns $\underline{h}_{\sum\limits_{j=0}^{i}a_j+l_1}$ and $\underline{h}_{\sum\limits_{j=0}^{i}a_j+l_2}$ in $H_1$ will make the resulting vector have a $0$ in $(\sum\limits_{l=0}^{i}\rho_l+j)^{th}$ co-ordinate and the resulting vector also has Hamming weight at most 2 with support contained in $\{\sum_{l=0}^{i-1}\rho_l+1,..,\sum_{l=0}^{i-1}\rho_l+\rho_i\}$ and hence applying Property $P_i$ on this resulting vector implies that there is a non-empty set of at most $2(i+1)+2$ columns in $H_1$ which are linearly dependent; hence contradicts the minimum distance as $2(i+1)+2 \leq t$; thus proving that each row of $D_{i+1}$ has Hamming weight exactly 1.
	\item Now we show that Property $P_{i+1}$ is true: We have to prove that any $m \times 1$ vector with Hamming weight at most $2$ with support contained in $\{\sum_{l=0}^{i}\rho_l+1,..,\sum_{l=0}^{i}\rho_l+\rho_{i+1}\}$ can be written as linear combination of at most $2(i+2)$ column vectors of $H_1$ say $\underline{h}_{p_1},...,\underline{h}_{p_\psi}$ for some $\{p_1,...,p_\psi\} \subseteq \{1,...,\sum_{l=0}^{i+1}a_l\}$ and $0 < \psi \leq 2(i+2)$. This can be easily seen using arguments similar to ones presented before.
	Let an $m \times 1$ vector have non-zero entries in co-ordinates $\sum_{l=0}^{i}\rho_l+j_1,\sum_{l=0}^{i}\rho_l+j_2$ or $\sum_{l=0}^{i}\rho_l+j_1$. Then this vector can be linearly combined with at most 2 column vectors in $H_1$ say $\underline{h}_{y_1},...,\underline{h}_{y_\phi}$ where $\{y_1,...,y_\phi\} \subseteq \{ \sum_{l=0}^{i}a_l+1,...,\sum_{l=0}^{i}a_l+a_{i+1} \}$ with $\phi = 1$ or $2$,  (2 columns $\underline{h}_{y_1},\underline{h}_{y_2}$ with first column and second column having a non-zero entry in co-ordinates $\sum_{l=0}^{i}\rho_l+j_1,\sum_{l=0}^{i}\rho_l+j_2$ respectively or a column $\underline{h}_{y_1}$ with a non-zero entry in $(\sum_{l=0}^{i}\rho_l+j_1)^{th}$ co-ordinate. These columns $\underline{h}_{y_1},\underline{h}_{y_2}$ or $\underline{h}_{y_1}$ exist due to $D_{i+1}$.) to form a $m \times 1$ vector with Hamming weight at most 2 with support contained in $\{ \sum_{l=0}^{i-1}\rho_l+1,\sum_{l=0}^{i-1}\rho_l +\rho_i \}$ which in turn can be written as linear combination of at most $2(i+1)$ column vectors in $H_1$ say $\underline{h}_{z_1},...,\underline{h}_{z_\theta}$ for some $\{z_1,...,z_\theta\} \subseteq \{1,...,\sum_{l=0}^{i}a_l\}$ and $0 < \theta \leq 2(i+1)$ by property $P_i$. Hence the given $m \times 1$ vector is written as linear combination of at most $2(i+2)$ column vectors in $H_1$ say $\underline{h}_{w_1},...,\underline{h}_{w_\gamma}$ for some $\{w_1,...,w_\gamma\} \subseteq \{1,...,\sum_{l=0}^{i+1}a_l\}$ and $0 < \gamma \leq 2(i+2)$.\\
\end{itemize}	
\bibliographystyle{IEEEtran}
\bibliography{bib_file}	

\end{document}